\documentclass[10pt,journal]{IEEEtran}

\usepackage{amsmath}
\usepackage[ruled,vlined,noend]{algorithm2e} 
\usepackage{subfigure} 
\usepackage{enumitem} 
\usepackage{amsfonts} 
\usepackage{amsthm}  
\usepackage{footnote} 

\usepackage{pgfplots} 
\pgfplotsset{compat=1.6}

\newtheorem{theorem}{Theorem}
\newtheorem{definition}{Definition}
\newtheorem{lemma}{Lemma}
\newtheorem{example}{Example}
\newtheorem{corollary}[theorem]{Corollary}
\newtheorem{property}{Property}

\newtheorem{remark}{Remark}

\newcommand{\ceil}[1]{\lceil #1 \rceil}
\newcommand{\floor}[1]{\lfloor #1 \rfloor}
\newcommand{\expect}[1]{\mathbb{E}[#1]}
\newcommand{\expectExplicit}[2]{\mathbb{E}_{#2}[#1]}
\newcommand{\norm}[1]{||#1||}
\newcommand{\randwalk}[2] {\ensuremath{RW{\left({#1},{#2}\right)}}}

\newcommand{\TRANSACT}[3] {\ensuremath{({#1} \to_{#2} {#3})}} 
\newcommand{\TRANSOLD}[3] {\ensuremath{({#1} \to_{2^{#2}} {#3})}}

\definecolor{myblue}{RGB}{80,80,160}
\definecolor{mygreen}{RGB}{80,160,80}

\newcommand{\newtext}[1]{\textcolor{black}{#1}} 

\begin{document}
\bstctlcite{IEEEexample:BSTcontrol} 

\title{Codes for Load Balancing in TCAMs: Size Analysis}


\author{Yaniv Sadeh, Ori Rottenstreich and Haim Kaplan}


\maketitle
{\let\thefootnote\relax\footnotetext{
This manuscript combines and extends two papers, presented at SOSR 2021~\cite{SadehSOSRTCAMsize} and at ISIT 2022~\cite{SadehISITTCAMsize}.
Yaniv Sadeh is with Tel-Aviv University, Israel (yanivsadeh@mail.tau.ac.il). 
Ori Rottenstreich is with the Technion - Israel Institute of Technology,  Israel  (or@technion.ac.il). 
Haim Kaplan is with Tel-Aviv University, Israel (haimk@tau.ac.il). }}

\begin{abstract}
Traffic splitting is a required functionality in networks, for example for load balancing over paths or servers, or by the source's access restrictions. The capacities of the servers (or the number of users with particular access restrictions) determine the sizes of the parts into which traffic should be split. A recent approach implements traffic splitting within the ternary content addressable memory (TCAM), which is often available in switches. It is important to reduce the amount of memory allocated for this task since TCAMs are power consuming and are often also required for other tasks such as classification and routing. Recent works suggested algorithms to compute a smallest implementation of a given partition in the longest prefix match (LPM) model. In this paper we analyze properties of such minimal representations and prove lower and upper bounds on their size. The upper bounds hold for general TCAMs, and we also prove an additional lower-bound for general TCAMs. We also analyze the expected size of a representation, for uniformly random ordered partitions. We show that the expected representation size of a random partition is at least half the size for the worst-case partition, and is linear in the number of parts and in the logarithm of the size of the  address space.
\end{abstract}

\section{Introduction}
\label{section_indtroduction}
In many networking applications, traffic has to be split into multiple possible targets. For example, this is required in order to partition traffic among multiple paths to a destination based on link capacities, and when sending traffic to one of multiple servers proportionally to their CPU or memory resources for load balancing. Traffic splitting also arises in maintaining access-control lists (ACLs). Here, we want to limit the number of users with specific permissions. We do this by associating a fixed quota of $W$-bit identifiers with each ACL, and granting a particular access only to users that have one of these identifiers.
In general, we address any scenario where traffic should be split by allocating a particular subset of identifiers of a specified size to each part (a part could be associated with a server or an ACL or with some other object).

It is increasingly common to rely on network switches  to perform the split~\cite{al2008scalable},\cite{Ananta},\cite{DASH_alg}. Equal cost multipath routing (ECMP)~\cite{RFC2992} uses hashing on flows to uniformly select one of target values written as memory entries. 
WCMP~\cite{WCMP},\cite{Zegura} (Weighted ECMP) generalizes the selection for non-uniform selections through entry repetitions, implying a distribution according to the number of appearances of each possible target. The implementation of some distributions in WCMP may require a large hash table. While for instance implementing a 1:2 ratio can be done with three entries (one for the first target and two for the second), the implementation of a  ratio  of the form $1:2^W-1$ is expensive, requiring $2^W$ entries. Memory can grow quickly for particular distributions over many targets, even if they are only being approximated. A recent approach~\cite{DASH_alg} refrains from memory blowup by comparing the hash to range-boundaries. Since the hash is tested sequentially against each range, it restricts the total number of load-balancing targets.

Recently, a natural approach was taken to implement traffic splitting within the Ternary Content Addressable Memory (TCAM), available in commodity switch architectures. For some partitions this allows a much cheaper representation~\cite{WangBR11},\cite{Niagara},\cite{AccurateExp},\cite{BitMatcher}. In particular, a partition of the form $1:2^W-1$ can be implemented with only two entries. Unfortunately, TCAMs are power consuming and thus are of limited size~\cite{appelman2012performance},\cite{McKeownABPPRST08}. Therefore a common goal is to minimize the representation of a partition in TCAMs. Finding a representation of a partition becomes more difficult when the number of possible targets is large. Focusing on the \emph{Longest Prefix Match model} (LPM), \cite{Niagara} suggested an algorithm named \emph{Niagara}, and showed that it produces small representations in practice.
They also evaluated a tradeoff of reduced accuracy for less rules. \cite{BitMatcher} suggested an optimal algorithm named \emph{Bit Matcher} that computes a smallest TCAM for a target partition. They also  proved that Niagara always computes  a smallest  TCAM, explaining its good empirical performance. A representation of a partition can be memory intensive when the number of possible targets is large. \cite{TCAM_Linf_TON} and \cite{TCAM_L1_INFOCOM} consider ways of finding approximate-partitions whose representation is cheaper than a desired input partition.

In this paper we analyze the size of the resulting (minimal) TCAM computed by Bit Matcher and Niagara. We prove upper and lower bounds on the size of the smallest TCAM needed to represent a given partition. We also give a lower bound on the minimum size of any TCAM\footnote{Henceforth, ``general TCAM'' or plain ``TCAM'' will refer to the general case, and ``LPM TCAM'' will refer to the restricted case.} needed to represent a given partition. Note that any upper bound for LPM TCAMs is an upper bound for general TCAMs. The optimization problem for general TCAMs, that is, how to find a smallest or approximately smallest TCAM for a given partition, is open. We also provide an average-case analysis to the size of an LPM TCAM.

{\bf Our Contributions.} 
(1) We prove two upper and two lower bounds on the size of the smallest LPM TCAM for a given partition over any number $k$ of targets.
The first upper and lower bounds are general and hold  for all partitions of $2^W$ into $k$ parts. They are derived through new analysis of the Bit Matcher algorithm \cite{BitMatcher}. The upper bound is roughly $\frac{1}{3}kW$. The two additional upper and lower bounds are partition-specific and consider the particular values $p_1,\ldots,p_k$ of the $k$ parts. The upper bound holds for LPM TCAMs. The lower bound has two versions, stronger for LPM TCAMs and weaker for general TCAMs. The partition-specific LPM bounds are a $2$-approximation. That is, the upper bound is at most twice larger than the lower bound.

(2) We demonstrate the tightness of our LPM bounds by constructing partitions with minimal TCAM representations that match them. We also provide examples showing that a general TCAM implementation for a partition may be strictly smaller than the best LPM implementation.

(3) We analyze the average size of the smallest LPM TCAM required to represent a partition drawn uniformly from the set of all ordered-partitions of $2^W$ into $k$ integer parts. This analysis shows that the expected size of the smallest LPM TCAM is roughly between $\frac{1}{6}kW$ and $\min(\frac{1}{5},\frac{1}{6} + \frac{1}{\sqrt{6\pi k}}) \cdot kW$, establishing that ``typical'' partitions cannot be encoded much more efficiently compared to the worst-case (factor of half).

(4) We evaluate how tight our bounds are, as well as our average-case analysis, experimentally, by sampling many partitions.

\newtext{The structure of the rest of the paper is as follows. In Section~\ref{section_model} we formally define the problem and set some terminology and definitions. In Section~\ref{section_size_by_bitmatcher} we derive general bounds that depend on the number of targets $k$ and the TCAM width $W$ and construct worst-case partitions that require a large number of rules. In Section~\ref{section_size_by_signed_bits} we derive additional bounds, tailored for any given partition, depending on the signed-bits representation of its parts. In Section~\ref{section_average_case} we analyze the average-case, i.e. the expected number of rules. We complement our analysis by experiments in Section~\ref{section_experiments}. Section~\ref{section_related_work} surveys some related work, and Section~\ref{section_conclusions} summarizes our findings and adds a few conclusions.}

\section{Model and Terminology}
\label{section_model}

A Ternary Content Addressable Memory (TCAM) of width $W$ is a table of entries, or \emph{rules}, each containing a \emph{pattern} and a \emph{target}. We assume that each target is an integer in $\{1,\ldots,k\}$, and also define a special target $0$ for dealing with addresses that are not matched by any rule. Each pattern is of length $W$ and consists of bits (0 or 1) and don't-cares ($*$). An \emph{address} is said to match a pattern if all of the specified bits of the pattern (ignoring don't-cares) agree with the corresponding bits of the address. If several rules fit an address, the first rule applies. An address $v$ is associated with the target of the rule that applies to $v$.

Most of the analysis in this paper follows the \emph{Longest Prefix Match (LPM)} model, restricting rule-patterns in the TCAM to include wildcards only as a suffix such that a pattern can be described by \emph{a prefix} of bits. 
This model is motivated by specialized hardware as in \cite{LPM_TCAM}, and is assumed in many previous studies \cite{WangBR11},\cite{Niagara},\cite{AccurateExp},\cite{BitMatcher}. Common programmable switch architectures such as RMT and Intel’s FlexPipe have tables  dedicated to LPM~\cite{NSDIJose15},\cite{Forwarding13},\cite{FlexPipe}. In general, much less is known about general TCAM rules. We do not have tighter upper bounds for general TCAMs, but we do provide a specific lower bound in Theorem~\ref{theorem_non_lpm_lower_bound_signed_bits}.

There can be multiple ways to represent the same partition in a TCAM, as a partition does not restrict the particular addresses mapped to each target but only their number. For instance, with $W=3$ the rules $\{ \textsc{011} \to 1, \textsc{01*} \to 2, \textsc{0**} \to 3, \textsc{***} \to 1\}$ imply a partition [5,1,2]  of addresses mapped to each of the targets \{1,2,3\}. Similarly, the same partition can also be derived using only three rules $\{\textsc{000} \to 2, \textsc{01*} \to 3, \textsc{***} \to 1\}$ (although this changes the identity of the addresses mapped to each target).

Given a desired partition $P = [p_1,\ldots,p_k]$ of the whole address space of $2^W$ addresses of $W$-bits such that $p_i > 0$ addresses should reach target $i$ ($\sum_i p_i = 2^W$), we aim to know the size of a smallest set of TCAM rules that partition traffic according to $P$. Note that $k \le 2^W$, and all addresses are considered equal in this model.\footnote{The model assumes implicitly that every address is equally likely to arrive, therefore the TCAM implementation only requires each target to receive a certain number of addresses. This assumption might not hold in practice, but it can be mitigated by ignoring bits which are mostly fixed like subnet masks etc. For example, \cite{Kang2014NiagaraSL} analyzed traces of real-data and concluded that for those traces about $6{-}8$ bits out of the client's IPv4 address are ``practically uniform''.}

A TCAM $T$ can be identified with a sequence $s$ of transactions between targets, defined as follows. Start with an empty sequence, and consider the change in the mapping defined by $T$ when we delete the first rule of $T$, with target $i \in \{1,\ldots,k\}$. Following this deletion some of the addresses may change their mapping to a different target, or become unallocated. If by deleting this rule, $m$ addresses are re-mapped from $i$ to $j \in \{0,\ldots,k\}$ (recall that $j=0$ means unallocated), we add to $s$ a transaction $\TRANSACT{i}{m}{j}$. We then delete the next rule of $T$ and add the corresponding transactions to $s$, and continue until $T$ is empty and all addresses are unallocated.

\begin{definition}[Transactions] 
\label{def_transaction}
Denote a transaction of size $m$ from $p_i$ to $p_j$ by $\TRANSACT{i}{m}{j}$. Applying this transaction to a partition $P = [p_1,\ldots,p_k]$ updates its values as: $p_i \leftarrow p_i - m$, $p_j \leftarrow p_j + m$.
\end{definition}

In the LPM model, a deletion of a single TCAM rule corresponds to exactly one transaction (or none if the rule was redundant), of size that is a power of $2$. 

\begin{example}
\label{example_tcam_to_sequence}
Consider the rules: $\{011 \to 1, 01{*} \to 2, 0{*}{*} \to 3, {*}{*}{*} \to 1 \}$ with $W=3$. They partition $2^W = 8$ addresses to $k=3$ targets. Deleting the first rule corresponds to the transaction $\TRANSACT{1}{1}{2}$. The deletion of each of the following three rules also corresponds  to a single transaction, $\TRANSACT{2}{2}{3}$, $\TRANSACT{3}{4}{1}$ and $\TRANSACT{1}{8}{0}$, respectively, see Fig.~\ref{figure_mapping_example}.

\begin{figure}[t]
  \centering
    \includegraphics[width=0.99\linewidth]{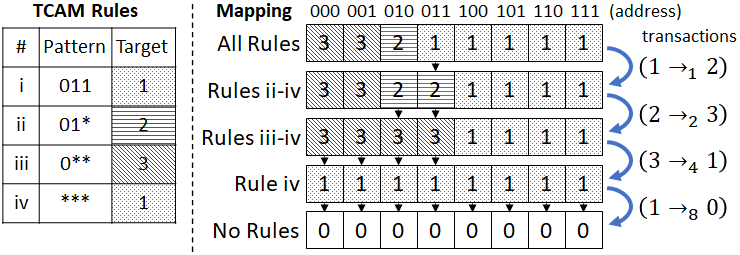}
  \caption{Example~\ref{example_tcam_to_sequence}: Rules and corresponding transactions due to remapping.}
  \label{figure_mapping_example}
\end{figure}
\end{example}

\begin{example}
\label{example_tcam_to_sequence_non_lpm}
To see that a deletion may correspond to multiple transactions, consider the following TCAM rules in order of priority: $\{0{*} \to 1, {*}0 \to 2, {*}1 \to 3 \}$. Deleting the first rule corresponds to the transactions $\TRANSACT{1}{1}{2}$ and $\TRANSACT{1}{1}{3}$.
\end{example}

\begin{figure}[t]
  \centering
    \includegraphics[width=0.8\linewidth]{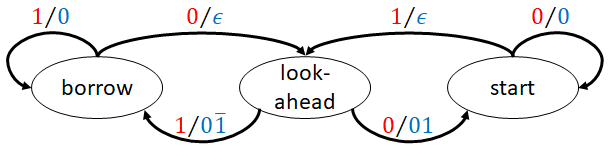}
  \caption{Automaton to convert binary to signed-bits form. Bits are in red, output signed-bits are in blue. To convert a number $n \ge 0$, process it with two leading zero bits. Reading and writing is done \emph{right-to-left} (LSB to MSB). For example, ${11}_{10} = {001011}_{2}$
  outputs $(01)\epsilon(0\overline{1})\epsilon(0\overline{1})\epsilon$,
  i.e.:
  ${11}_{10} = 10\overline{1}0\overline{1}$.}
  \label{figure_signed_bits_automaton}
\end{figure}

\begin{definition}[Complexity] 
\label{definition_length_of_partition}
Let $P=[p_1,\ldots,p_k]$ be a partition of $2^W$. We define $n(P)$ to be the size of the smallest general TCAM that realizes $P$. We also define $\lambda(P)$ to be the length of a shortest sequence of transactions of sizes that are powers of $2$, that zeroes $P$. This is also equal to the size of a smallest LPM TCAM that realizes $P$.
We say that $\lambda(P)$ is the \emph{complexity} of $P$. 
\end{definition}

It was shown in~\cite{BitMatcher} that Bit Matcher (see Algorithm~\ref{alg_match_variants})
computes a (shortest) sequence of transactions for an input partition $P$ whose length is $\lambda(P)$ which can also be mapped to a TCAM table. No smaller LPM TCAM exists, as it will correspond to a shorter sequence in contradiction to the minimality of $\lambda(P)$.

We conclude this section with important definitions, Tables~\ref{table_all_results_bounds}-\ref{table_all_results_existence} that summarize all the results from subsequent sections, and a remark regarding general TCAMs.

\begin{algorithm}[t]
    \SetAlgoLined
    \DontPrintSemicolon

    \SetKwFunction{funcGenerateSequence}{Compute}
    \SetKwFunction{funcProcessLowLevel}{ProcessLevel}
    \SetKwProg{Fn}{Function}{:}{}
    
    {
    \Fn{\funcGenerateSequence{partition $P$, $\sum_{i=1}^k p_i =2^W$}}{
        Initialize $s$ to be an empty sequence.\;
        \For{level $d = 0 \ldots W-1$} {
            $s' = \funcProcessLowLevel(P,d)$. Update $s \leftarrow s \cup s'$ and apply $s'$ on $P$.
        }
        \Return $s \cup \{ \TRANSACT{i}{2^W}{0} \ |\ 1 \le i \le k \wedge p_i = 2^W \}$.
        }
    \;
    // \textbf{Bit Matcher (BM) \cite{BitMatcher}} processing \;
    \Fn{\funcProcessLowLevel{partition $P$, $d$}}{
        {\footnotesize // $p_i$ is \emph{bit lexicographic} smaller than $p_j$ if ${p_i}^r < {p_j}^r$ where $n^r$ is the bit-reverse of $n$ with respect to word size $W$.}
    
       Let $A = \{ i \mid i \ge 1 \wedge p_i[d] = 1\}$.\;
       
       Let $A_h \subset A$ consists of the $|A|/2$ indices of the $p_i$s that are largest in bit lexicographic order, and let $A_l = A \setminus A_h$. Pair the elements of $A_h$ and $A_l$ arbitrarily. For each pair $i \in A_l,j \in A_h$ append to $s'$ (initially $s' = \emptyset$) the transaction $\TRANSACT{i}{2^d}{j}$.\;
       Finally, \Return $s'$.
       
    }
    \;
    // \textbf{Random Matcher (RM)} processing \;
    \Fn{\funcProcessLowLevel{partition $P$, $d$}}{
        Let $s' = \emptyset$. Pair the weights $i \ge 1$ where $p_i[d] = 1$ uniformly at random. For each pair, say $i,j$: \;
        If $p_i[d+1] < p_j[d+1]$: Append $\TRANSACT{i}{2^d}{j}$ to $s'$. \;
        If $p_i[d+1] > p_j[d+1]$: Append $\TRANSACT{j}{2^d}{i}$ to $s'$. \;
        If $p_i[d+1] = p_j[d+1]$: Append to $s'$ either $\TRANSACT{i}{2^d}{j}$ or $\TRANSACT{j}{2^d}{i}$ with equal probability. \;
        Finally, \Return $s'$.
    }
    \;
    // \textbf{Signed Matcher (SM)} processing \;
    \Fn{\funcProcessLowLevel{partition $P$, $d$}}{
        Let $s' = \emptyset$. Let $F = \{ i\ |\ \phi(p_i)[d] = 1 \}$ and $G = \{ i\ |\ \phi(p_i)[d] = -1 \}$. Pair $F$ and $G$, and for each such pair $i \in F$,$j \in G$ append $\TRANSACT{i}{2^d}{j}$ to $s'$. If $|F| > |G|$, for each unpaired $i \in F$, append $\TRANSACT{i}{2^d}{0}$ to $s'$. If $|F| < |G|$, for each unpaired $j \in G$, append $\TRANSACT{0}{2^d}{j}$ to $s'$. Finally, \Return $s'$.
    }
    }
    \caption{Matching Variants: Bit Matcher (BM), Random Matcher (RM), Signed Matcher (SM)}
    \label{alg_match_variants}
\end{algorithm}

\begin{definition}[Expected Complexity]
\label{definition_avg_length_of_partition}
We define the \emph{normalized expected complexity} as: $L(k,W) \equiv \frac{\mathbb{E}[\lambda(P)]}{Wk}$ where $P$ is a uniformly random ordered-partition of $2^W$ to $k$ positive parts.\footnote{Two partitions $P$ and $P'$ with the same components ordered differently are considered different ordered-partitions. For example $[1,3] \ne [3,1]$.} $L(k,W)$ can be interpreted as an ``average rules per bit'' since the partition has $k$ numbers, each of $W$ bits. We also define $L(k) \equiv \lim_{W \to \infty}{L(k,W)}$.
\end{definition}

\begin{definition}[Signed-bits Representation]
\label{definition_signed_bits}
The \emph{canonical signed-bits representation} of a number $n$, denoted by $\phi(n)$, is a representation such that $n = \sum_{i=0}^{d}{(2^i a_i)}$ where $\forall i: a_i \in \{-1,0,1\}$, $a_d \ne 0$, and $\forall i: a_i \cdot a_{i+1} = 0$ (no consecutive non-zeros). We denote by $|\phi(n)|$ the number of \textbf{non-zero} signed-bits in this representation,\footnote{The sequence $|\phi(n)|$ for $n \ge 0$ is known as https://oeis.org/A007302.}
and extend $\phi$ to vectors $P = [p_1,\ldots,p_k]$ element-wise: $|\phi (P)| = \sum_{i=1}^{k}{|\phi(p_i)|}$. Also, let $M(P) = \max_{i \in [k]}{|\phi(p_i)|}$. Finally, we denote $-1$ by $\overline{1}$.
\end{definition}

One can verify that $\phi(n) = n$ for $n \in \{-1,0,1\}$, and: $\phi(2n) = \phi(n) \circ 0$, $\phi(4n+1) = \phi(n) \circ 01$, $\phi(4n-1) = \phi(n) \circ 0\overline{1}$ where $\circ$ stands for concatenation. This representation can also be computed from least to most significant bit by the automaton in Fig.~\ref{figure_signed_bits_automaton}. It converts long sequences of $1$-bits to $0$, and determines whether each sequence should begin (least significant bit) with $1$ or $\overline{1}$ according to its state. \cite{signedBitsAutomaton} studies a more general problem and also discusses this particular automaton. Overall, signed-bits are of interest since they come up in minimization/optimization scenarios, see \cite[Section 6]{SignedBitsLooplessGrayCode} for a survey.
One may think of a signed representation  of an integer as the difference of two non-negative integers.

\begin{definition}[Level]
\label{definition_binary_levels}
Let $n$ be a non-negative integer. We denote by $n[\ell]$ the $\ell$-th bit in the binary representation of $n$. We denote by $\phi(n)[\ell]$ the $\ell$-th signed-bit of $\phi(n)$. We refer to $\ell$ as the level of this bit, $\ell=0$ is the least-significant level.
\end{definition}

\begin{table}[!ht]
    \begin{center}
        \begin{tabular}{|c|c|}

        \hline
        Remark~\ref{remark_size_lower_bound_k} &
        $\lambda(P) \ge n(P) \ge k$ \\
        \hline
        
        Theorem~\ref{theorem_size_upper_bound_third} &
        \begin{tabular}{@{}c@{}}
            $k = 2:   \lambda(P) \le \frac{1}{2}W + 2$ \\
            $k \ge 3: \lambda(P) \le \frac{1}{3} k(W - \floor{\lg k} + 4)$
        \end{tabular} \\
        \hline
        
        Theorem~\ref{theorem_bounds_signed_bits} &
        $\ceil{\frac{|\phi(P)|+1}{2}} \le \lambda(P) \le |\phi(P)| +1-M(P)$ \\
        \hline
        
        Theorem~\ref{theorem_non_lpm_lower_bound_signed_bits} & 
        \begin{tabular}{@{}c@{}}
            Assume $|\phi(p_1)| \ge |\phi(p_2)| \ge \ldots \ge |\phi(p_k)|$, then \\
            $n(P) \ge \max_{i=1,\ldots,k}{\lg ( |\phi(p_i)| + 1) + i - 1}$
        \end{tabular} \\
        \hline
        \hline
        
        Theorem~\ref{theorem_avg_case_rand_bits_rough} &
        $L(k) \in [\frac{1}{6},\frac{1}{5}]$ and $L(2) = \frac{1}{6}$ \\
        \hline
        
        Theorem~\ref{theorem_bound_per_k} &
        $L(k) \le \frac{1}{6} + c(k)$\ , $\lim_{k \to \infty}{\sqrt{6\pi k} \cdot c(k)} = 1$ \\
        \hline
        
        \end{tabular}
    \end{center}
    \caption{Summary of bounds regarding $n(P)$, $\lambda(P)$
    and $L(k)$ in terms of the parameters $k$ and $W$, or
    in terms of signed-representation notations $\phi(p_i)$, $\phi(P)$, $M(P)$ as in Definitions~\ref{definition_signed_bits}.}
    \label{table_all_results_bounds}
\end{table}

\begin{table}[!ht]
    \begin{center}
        \begin{tabular}{|c|c|}

        \hline
        
        Theorem~\ref{theorem_example_worst_case_k_2} &
        $k=2: \exists P \text{\ s.t.\ } \lambda(P) = \ceil{\frac{W}{2}} + 1$ \\
        \hline
        
        Theorem~\ref{theorem_example_worst_case_k_3} &
        $k=3: \exists P \text{\ s.t.\ } \lambda(P) = W+1$ \\
        \hline
        
        Theorem~\ref{theorem_example_worst_case_k_any} &
        $k \ge 4: \exists P \text{\ s.t.\ } \lambda(P) > \floor{\frac{k-1}{3}}(W-\ceil{\lg k}+1)$ \\
        \hline
        
        Theorem~\ref{theorem_tight_signed_bits_bounds} &
        $\exists P$ s.t. Theorem~\ref{theorem_bounds_signed_bits} is tight (UB/LB/both) \\
        \hline
        
        Corollary~\ref{corollary_lower_bound_non_lpm_size} &
        $\exists P$ s.t.\ $n(P) \ge \lg (W - \ceil{\lg k} + 3) + k-2$ \\
        \hline
        
        \end{tabular}
    \end{center}
    \caption{Summary of results regarding construction of partitions to demonstrate the tightness of the bounds in Table~\ref{table_all_results_bounds}.}
    \label{table_all_results_existence}
\end{table}

\begin{remark}[General TCAMs]
\label{remark_non_prefix_remark}
No non-trivial algorithm to compute exactly or to approximate a smallest general TCAM for a given partition is known.\footnote{If the mapping is a function, i.e. each address has a predefined target, then there are exponential algorithms, see \cite{tcams_function_non_lpm2009}.} On the other hand, there are also no hardness proofs for this problem. The LPM model is much simpler because every rule corresponds to exactly one transaction. 
Indeed, using general rules sometimes allows a smaller TCAM. For the following two partitions $\lambda(P) > n(P)$.

The partition $P = [4,3,3,3,3]$ satisfies $\lambda(P) = 7$. But $n(P) = 5$ as follows: $ \{ {*}{*}{0}{0} \to 1 , {0}{0}{*}{*} \to 2 , {0}{1}{*}{*} \to 3 , {1}{0}{*}{*} \to 4 , {1}{1}{*}{*} \to 5 \}$. Deleting the first rule corresponds to $4$ transactions (one-to-many): $\TRANSACT{1}{1}{i}$ for $i=2,3,4,5$.

Using general rules can reduce the size of a TCAM even for $k=2$. For $W=10$, to represent $P = [683,341]$ we need  $\lambda(P) = 6$ LPM rules (see Theorem~\ref{theorem_example_worst_case_k_2}). But $n(P) \le 5$:
$ \{
{0}{0}{0}{0}{0}{0}{0}{0}{0}{0} \to 2 ,
{*}{0}{0}{0}{*}{*}{*}{0}{0}{0} \to 1 ,
{*}{*}{0}{0}{0}{*}{*}{*}{*}{*} \to 2 ,
{0}{0}{*}{*}{*}{*}{*}{*}{*}{*} \to 2 ,
{*}{*}{*}{*}{*}{*}{*}{*}{*}{*} \to 1
\}$.
\end{remark}



\section{TCAM Size Bounds in terms of $k$ and $W$}
\label{section_size_by_bitmatcher}

In this section we prove upper and lower bounds on the complexity $\lambda(P)$ of a  partition $P$ as a function of the number of targets $k$ and the sum $2^W$ (TCAM width $W$).


\begin{remark}[Trivial Lower Bound]
\label{remark_size_lower_bound_k}
For any partition $P$ to $k$ targets we have that $n(P) \ge k$. This is because each target must be associated with at least one rule. Also, $\lambda(P) \ge n(P)$ because LPM rules are more restrictive.
\end{remark}

The remaining analysis in this section focuses on the upper bound, and is based on the properties of sequences 
which compute an optimal (minimal size) LPM TCAM for a given partition as generated by the Bit Matcher algorithm~\cite{BitMatcher} described in Algorithm~\ref{alg_match_variants}. Observe that the transactions are generated in an increasing order of size. The analysis relies on the bit-lexicographic order \newtext{(see Algorithm~\ref{alg_match_variants})} used by Bit Matcher.



The first lemma refers to changes in the binary representation of the  weights following  an application  of a transaction.

\begin{lemma}
\label{lemma_zeroed_bits_strict}
Let $P = [p_1,\ldots,p_k]$ be a partition of $2^W$, and let $s$ be a sequence of transactions generated for $P$ by Bit Matcher (see Algorithm~\ref{alg_match_variants}). The  following two properties hold: 

(i) Consider a specific transaction $\TRANSACT{i}{2^\ell}{j}$ after all smaller transactions have been applied. Let $z$ denote the number of consecutive bits in levels $\ge \ell$
of $p_i$ and $p_j$ that are guaranteed to be zero (and stay zero) following this transaction.\footnote{Note that levels $0,\ldots,\ell-1$ are known to be zero \emph{prior} to the transaction.} If $\ell < W-1$, then $z \ge 3$. Moreover, if $k=2$ then $z \ge 4$. 

(ii) Zeroed bits associated with different transactions are different, and the bits that are zeroed by $\TRANSACT{i}{2^\ell}{j}$ are consecutive from $p_i[\ell]$ and $p_j[\ell]$ upwards, until the next level in which $i$ and $j$ participate in a transaction (this level may be different for $i$ and $j$).
\end{lemma}

\begin{example}
Assume that $p_i = 5$ and $p_j = 3$ and that we apply the transaction $\TRANSACT{i}{1}{j}$. Then the weights change to $p_i = p_j = 4$, whose binary representation is $1\underline{00}$.
The zeroed bits are underlined, and in this example, $z=4$ (two in $i$ and two in $j$).
\end{example}

\begin{proof}[Proof of Lemma~\ref{lemma_zeroed_bits_strict}]
Bit Matcher applies a transaction at level $\ell$  between targets $i$ and $j$ with
bit $\ell$ equals 
$1$. It follows that the first two bits that are guaranteed to be zero  following this transaction are $p_i[\ell]$ and $p_j[\ell]$.

Since Bit Matcher applies the transaction in bit-lexicographic order, before the transaction is applied we have $p_i[\ell+1] \le p_j[\ell+1]$. Consider the three possible cases:
\begin{enumerate}[leftmargin=*,label=(\arabic*)]
    \item $p_i[\ell+1]=0$ and $p_j[\ell+1]=0$: following the transaction, we have $p_i[\ell+1]=0$, $p_j[\ell+1]=1$. 
    
    \item \label{case_2z} $p_i[\ell+1]=0$ and $p_j[\ell+1]=1$: following the transaction, we have $p_i[\ell+1]=0$, $p_j[\ell+1]=0$ (due to carry). 
    
    \item $p_i[\ell+1]=1$ and $p_j[\ell+1]=1$: following the transaction, we have $p_i[\ell+1]=1$, $p_j[\ell+1]=0$ (due to carry). 
\end{enumerate}

Overall, we see that $z \ge 3$ ($p_i[\ell] = p_j[\ell] = 0$ and at least another bit in level $\ell+1$). When $k=2$, only Case~\ref{case_2z} is possible, because $p_1 + p_2 = 2^W$, so we get $z \ge 4$.
Property (ii) of the claim is trivial by the way Bit Matcher works.
\end{proof}

Now we use Lemma \ref{lemma_zeroed_bits_strict} to prove a worst-case upper bound. That is, the largest possible minimum-size LPM TCAM for a partition of $2^W$ addresses to $k$ targets.

\begin{theorem}[Upper Bound]
\label{theorem_size_upper_bound_third}
Let $P$ be a partition of $2^W$ to $k$ parts.
\begin{itemize}[leftmargin=*,noitemsep]
    \item If $k=2$: $\lambda(P) \le \frac{1}{4} k(W - \floor{\lg k} + 1) + k = \frac{1}{2}W + 2$.
    \item If $k \ge 3$: $\lambda(P) \le \frac{1}{3} k(W - \floor{\lg k} + 1) + k = \frac{1}{3} k(W - \floor{\lg k} + 4)$.
\end{itemize} 
\end{theorem}

\begin{proof}
We argue first for $k \ge 3$ and then consider the case $k=2$. Consider the binary representation of the weights $p_1,\ldots,p_k$. By Lemma~\ref{lemma_zeroed_bits_strict}, 
each  transaction of Bit Matcher at level $\ell$ can be associated with (at least) $3$ bits at levels $\ell$ and $\ell+1$ that remain zero after it is applied. Furthermore, bits associated with different transactions are different.

Let $N = W - \floor{\lg k}$. Since there are at most $k(N + 1)$ bits in the $N + 1$ least significant levels (in total for all targets),
and we associated uniquely three bits with each transaction, then we can have at most $\frac{1}{3} k(N + 1)$ transactions at the $N$ least significant levels. Following these transactions the $N$ least significant levels of $p_1,\ldots,p_k$ are all zero. We considered $N+1$ levels since transactions at the first $N$ levels may ``charge'' bits at the $N+1$ level.

After these transactions have been applied, we observe that no more than $k$ bits are $1$ in the (current) binary representations of all $p_i$, for $i \ge 1$. Indeed, each $1$ bit at the top $\floor{\lg k}$ levels
contributes at least $2^{W-\floor{\lg k}} \ge \frac{2^W}{k}$ to the sum of all $p_i$, which equals $2^W$. Thus,
Bit Matcher performs at most $k$ additional transactions at the most significant $\floor{\lg k}$ levels. In total we get that for any partition $P$ with $k \ge 3$: $\lambda(P) \le \frac{1}{3} k(W - \floor{\lg k} + 1) + k$.

When $k=2$, we repeat the same argument. The only difference is that each transaction at level $\le N$ is associated with (at least) $4$ zeroed bits (by Lemma~\ref{lemma_zeroed_bits_strict}), so the factor of $\frac{1}{3}$ in the bound changes to $\frac{1}{4}$.
\end{proof}

Next, we show worst-case partitions matching the upper bounds on $\lambda(P)$ in Theorem~\ref{theorem_size_upper_bound_third} up to minor additive constants.

\begin{theorem}[$k=2$]
\label{theorem_example_worst_case_k_2}
Let $x = round\left(\frac{2^W}{3}\right)$ for $W \ge 1$. The partition $P = [x,2^W-x]$ satisfies $\lambda(P) = \ceil{\frac{W}{2}} + 1$.
\end{theorem}

\begin{proof}
The binary expansion of $\frac{1}{3} = 0.0101\ldots$ is infinite with alternation of $0$s and $1$s. The binary representations of $\frac{2^W}{3}$ and $\frac{2 \cdot 2^W}{3}$ are both shifts of this representation, so modulo $2$ one of them is $0.1010\ldots$ and the other is $1.0101\ldots$. In both cases the rounding makes both $x$ and $2^W-x$ odd, so both $x$ and $2^W - x$ have alternating and complementing bit representations except for the least significant bit in which both have $1$.
The Bit Matcher algorithm applies $\ceil{\frac{W}{2}}$ transactions, at every even level ($\ell = 0,2,\ldots$). It follows that the total number of transactions is $\ceil{\frac{W}{2}} + 1$, the $+1$ is due to the last transaction $\TRANSACT{i}{2^W}{0}$ for $i=1$ or $i=2$.
\end{proof}

\begin{theorem}[$k=3$]
\label{theorem_example_worst_case_k_3}
Let $x = \floor{\frac{2^W}{3}}$. If $x$ is even, let $P = [x,x+1,x+1]$, and if $x$ is odd let $P = [x,x,x+1]$ (one can verify that $\sum_{i=1}^{3}{p_i} = 2^W$ in both cases). Then $\lambda(P) = W+1$.
\end{theorem}

Note that if we substitute $k=3$ in Theorem~\ref{theorem_size_upper_bound_third} we get an upper bound of $W+3$.

\begin{proof}
\newtext{Denote the number of transactions for the partition $[x,x,x+1]$ as $\lambda_1(W)$ and for $[x,x+1,x+1]$ as $\lambda_2(W)$.}

\newtext{Consider first the case of odd $x$ and the partition $[x,x,x+1]$. Since $x$ is odd, the first transaction of Bit Matcher matches the two $x$-weights, and following this step we get the partition $[x-1,x+1,x+1]$. The resulting weights are all even, and therefore equivalent to the partition $[\frac{x-1}{2},\frac{x+1}{2},\frac{x+1}{2}]$: If $s$ is a sequence that corresponds to $[\frac{x-1}{2},\frac{x+1}{2},\frac{x+1}{2}]$ then if we multiply the size of each transaction of $s$ by $2$ we get a sequence that corresponds to $[x-1,x+1,x+1]$, and vice versa. If we denote $y = \frac{x-1}{2}$ the partition $[\frac{x-1}{2},\frac{x+1}{2},\frac{x+1}{2}]$ becomes $[y,y+1,y+1]$, that sums to $2^{W-1}$, so we conclude that $\lambda_1(W) = 1 + \lambda_2(W-1)$.}
    
\newtext{Now consider the case of $[x,x+1,x+1]$. In this case $x+1$ is odd, so after the first transaction of Bit Matcher we get: $[x,x,x+2]$, and similarly we conclude that $\lambda_2(W) = 1 + \lambda_1(W-1)$.}

\newtext{The relations $\lambda_1(W) = 1 + \lambda_2(W-1)$ and $\lambda_2(W) = 1 + \lambda_1(W-1)$
 simply imply that we have a transaction per level. The base-case for this recursive relation is $\lambda_1(0) = \lambda([0,0,1]) = 1$. We conclude that $\lambda(P) = W+1$.}
\end{proof}

\begin{theorem}[$k>3$]
\label{theorem_example_worst_case_k_any}
There exists a partition $P$ such that $\lambda(P) > \floor{\frac{k-1}{3}}(W-\ceil{\lg k}+1)$.
\end{theorem}

\begin{proof}
\newtext{The proof is very technical, but its idea is simple:} Divide the parts to $m = \floor{\frac{k-1}{3}}$ disjoint triplets and allocate a total weight of $2^{W-1 - \ceil{\lg m}}$ to each triplet as in Theorem~\ref{theorem_example_worst_case_k_3}. Allocate the remaining weight to the remaining parts. The core of the proof is to show that each triplet contributes $W- \ceil{\lg m} > W-\ceil{\lg k}+1$ transactions.

\newtext{Formally, let $m = \floor{\frac{k-1}{3}}$ be the number of triplets. This leaves out $k - 3m$ left-out targets, which is a number between 1 and 3. Now, for each triplet we allocate a total weight of $2^{W-1 - \ceil{\lg m}}$, and divide it between the targets of the triplet as in Theorem~\ref{theorem_example_worst_case_k_3}. This assignment allocates $m \cdot 2^{W-1 - \ceil{\lg m}}$ out of the total $2^W$ weight, so a total-weight of $N = 2^W - m \cdot 2^{W-1 - \ceil{\lg m}}$ remains to be divided between the left-out targets. Note that $N>0$ because the total weight of all the triplets is at most $2^{W-1}$.}

\looseness=-1
\newtext{If we can argue that a shortest sequence of transactions exists such that the targets of each triplet perform transactions between themselves, without interference from targets of other triplets and the left-out targets, then by Theorem~\ref{theorem_example_worst_case_k_3} we get at least $(W-1 - \ceil{\lg m}) + 1$ transactions per triplet, which can be lower-bounded by:
$$
W - \ceil{\lg m}
\ge W - \ceil{\lg \frac{k-1}{3}}
= W - \ceil{\lg k + \lg \frac{k-1}{3k}}
$$
$$
\ge W - \ceil{\lg k - 1.58}
\ge W - \ceil{\lg k} + 1
$$
We sum this up over the $m$ triplets and get a lower-bound of $\floor{\frac{k-1}{3}} (W - \ceil{\lg k} + 1)$ on the total number of rules. The true lower-bound is larger by $k-3m$, because each of the left-out targets must participate in at least one transaction, but this addition is at most $3$, so we neglect it for the sake of simplicity of the lower-bound expression.}

\newtext{Now, it remains to explain how to divide the left-out total-weight $N$ among the left-out targets, and why we can assume that effectively it is as if the transactions happen separately for each triplet. First, note that by symmetry we may assume that transactions don't happen between targets that belong to different triplets. It remains to explain why the left-out targets do not break this symmetry:}
\begin{itemize}[leftmargin=*,noitemsep]
    \item \newtext{If there is a single left-out target: then its value must be $N$. Since $N = 2^W - m \cdot 2^{W-1 - \ceil{\lg m}}$ it must be  a non-zero multiple of $2^{W-1 - \ceil{\lg m}}$. This means that if we consider a Bit Matcher sequence, the left-out target will participate in a transaction only at a level $W-1 - \ceil{\lg m}$ or higher, which is higher than the levels in which the triplets have their transactions. Thus, at the point in time where the left-out target starts making transactions, we have already made $\floor{\frac{k-1}{3}} (W - \ceil{\lg k})$ transactions. At this point, additional $m+1$ are guaranteed for $m+1$ non-zero weights.}
    
    \item \newtext{If there are two left-out targets: we partition the total left-out weight such that one target is of size $N-1$ and the other is of size $1$. Since $N-1$ is initially the largest in $<_{lex}$ order (it is a large power of 2 minus 1), and $1$ is the smallest, we may assume that Bit Matcher matches the two left-out targets together, so they do not interfere with the triplets. Afterwards, the same arguments of the previous case apply.}
    
    \item \newtext{If there are three left-out targets: they form their own triplet, but we cannot simply allocate them like the other triplets unless $N$ happens to be a power of $2$. So instead we partition the total left-out weight such that one target is of size $\frac{N}{2}$, another is of size $\frac{N}{2} - 1$ and the third is of size $1$.
    By the same argument as the previous case, we may assume that Bit Matcher matches $\frac{N}{2}-1$ with $1$, without interfering with the transaction of each triplet, so after the first level
    the weights of the left-out targets become $\frac{N}{2}$, $\frac{N}{2}$ and $0$.
    Since $N$ is a multiple of $2^{W-1 - \ceil{\lg m}}$, $\frac{N}{2}$ is a multiple of $2^{W-2 - \ceil{\lg m}}$, and the lowest level at which the left-out weights will participate in a transaction is $W-2 - \ceil{\lg m}$. When the sequence of transactions is executed up to that level,
    the weights of the targets in each triplet are $2^{W-2 - \ceil{\lg m}}$, $2^{W-2 - \ceil{\lg m}}$ and $0$, and each triplet already produced $W- \ceil{\lg m} - 2$ transactions. 
    Since two more targets of each triplet are yet to become zero, and each transaction zeroes at most one target we can associate two more transaction with each triplet. So the total number of transactions is still guaranteed to be as stated.}
\end{itemize}

\newtext{In conclusion, this construction requires at least $\floor{\frac{k-1}{3}}(W - \ceil{\lg k} + 1)$ rules, which is only slightly smaller than the upper bound $\frac{1}{3} k(W - \floor{\lg k} + 4)$ of Theorem~\ref{theorem_size_upper_bound_third}.}
\end{proof}

\begin{example}[\newtext{Concrete Example for Theorem~\ref{theorem_example_worst_case_k_any}}]
\label{example_for_upper_bound_construction}
\newtext{The following example illustrates the last case of Theorem~\ref{theorem_example_worst_case_k_any}.}

\newtext{Let $k=12$ and $W = 7$. So we divide the targets to $m=\floor{\frac{k-1}{3}} = 3$ triplets, and have $3$ additional left-out targets. We allocate the sum of each triplet to be $2^{W-1-\ceil{\lg 3}} = 2^4 = 16$, and the total left-out weight is $N = 128 - 3 \cdot 16 = 80$. In each triplet, the weights are divided as close to $\frac{16}{3}$ as possible, which means $5$, $5$ and $6$. There are three left-out targets, so the left-out weight is divided as $1$, $39$ and $40$. Thus, the resulting partition is: $P = [5,5,6,5,5,6,5,5,6,1,39,40]$.}

\newtext{One possible Bit Matcher sequence for $P$ is as follows:}
\begin{itemize}[leftmargin=*,noitemsep,topsep=0pt]
    \item \newtext{First transactions, level-$0$, transactions from first to second target in each triplet: \TRANSOLD{1}{0}{2} \TRANSOLD{4}{0}{5} \TRANSOLD{7}{0}{8} \TRANSOLD{9}{0}{10} after which, the state is $[4,6,6,4,6,6,4,6,6,0,40,40]$}

    \item \newtext{Level-$1$, transactions from second to third target in each triplet: \TRANSOLD{2}{1}{3} \TRANSOLD{5}{1}{6} \TRANSOLD{8}{1}{9} after which, the state is $[4,4,8,4,4,8,4,4,8,0,40,40]$}

    \item \newtext{Level-$2$, transactions from first to second target in each triplet: \TRANSOLD{1}{2}{2} \TRANSOLD{4}{2}{5} \TRANSOLD{7}{2}{8} after which, the state is $[0,8,8,0,8,8,0,8,8,0,40,40]$}

    \item \newtext{From this point, when one of the weights in each triplet became zero, the targets of the triplets might have transactions with the left-out targets. For example, Bit Matcher can make the following transaction at level 3: \TRANSOLD{2}{3}{3} \TRANSOLD{5}{3}{6} \TRANSOLD{8}{3}{11} \TRANSOLD{9}{3}{12} after which the state is $[0,0,16,0,0,16,0,0,0,0,48,48]$. Observe that $8$ and $9$ belong to the same triplet, yet both make a transaction with a left-out target.}

    \item \newtext{The sequence concludes with: \TRANSOLD{3}{4}{11} \TRANSOLD{6}{4}{12} \TRANSOLD{11}{6}{12} \TRANSOLD{12}{7}{0}}
\end{itemize}

\newtext{Overall, we have $18$ transaction, compared to the lower-bound of (substitute $k=12$,$W=7$): $\floor{\frac{12-1}{3}}(7 - \ceil{\lg 12} + 1) = 12$. If we do not neglect the addition of $k-3m$ transactions due to the left-over weights, the lower bound is in fact $15$ for this example.}
\end{example}

\section{TCAM Size Signed-Bits Bounds}
\label{section_size_by_signed_bits}

In this section we prove tighter lower and upper bounds on $\lambda(P)$. The upper bound applies to $n(P)$ as well, and we also improve the trivial lower bound of $k$ on $n(P)$. These bounds depend on the signed-bit representation of $p_1,\ldots,p_k$, rather than just on $k$ and $W$. \newtext{Revisit Definition~\ref{definition_signed_bits} for notations.}

\begin{property}[Signed-bits Sparsity]
\label{property_signed_bits_most_sparse}
\newtext{Proven in \cite[Section 4]{SignedBitsChords}: Let $\norm{d}$ denote the number of $1$ bits in the binary representation of $d$. If $n = x-y$ for non-negative integers $x,y$, then $|\phi(n)| \le \norm{x} + \norm{y}$. Equality is obtained when $x = n_{+} \equiv \sum_{a_i = 1}{2^i}$ and $y = n_{-} \equiv \sum_{a_i = \overline{1}}{2^i}$ where $a_i$ are the coefficients in $\phi(n)$.}
\end{property}

\newtext{Regarding Property~\ref{property_signed_bits_most_sparse}, we note that the terminology of \cite{SignedBitsChords} uses regular expressions. The idea behind the statement in "bitwise terminology" is to show that any $x$ and $y$ such that $n=x-y $ and $\norm{x} + \norm{y}$ is minimized, can be modified by repeated  changes $(x,y) \to (x',y')$ from the LSB upwards, such that we end-up with $(n_{+},n_{-})$ and $\norm{n_{+}} + \norm{n_{-}} \le \norm{x} + \norm{y}$ (hence, minimality follows).}

\begin{lemma}
\label{lemma_signed_bits_cancellation}
For  integers $n$ and $h$: $\big | |\phi(n + 2^h)| - |\phi(n)| \big | \le 1$.
\end{lemma}

\begin{proof}
For every integer $k$ we use the notations of $k_{+}$ and $k_{-}$ as defined in Property~\ref{property_signed_bits_most_sparse}. Denote $m = n + 2^h$. Then $n = m - 2^h = m_{+} - (m_{-} + 2^h)$, and therefore, together with Property~\ref{property_signed_bits_most_sparse}: $|\phi(n)| \le \norm{m_{+}} + \norm{m_{-} + 2^h} \le \norm{m_{+}} + \norm{m_{-}} + \norm{2^h} = |\phi(m)| + 1$, where the second inequality is by counting "standard" bits (it may be a strict inequality due to potential carry). The other direction is similar: $|\phi(m)| \le \norm{n_{+} + 2^h} + \norm{n_{-}} \le |\phi(n)| + 1$. Together: $\big | |\phi(n + 2^h)| - |\phi(n)| \big | \le 1$.
\end{proof}

\begin{theorem}
\label{theorem_bounds_signed_bits}
For any partition $P$, $\left\lceil \frac{|\phi(P)|+1}{2} \right\rceil \le \lambda(P) \le |\phi(P)|+1 - M(P)$.
\end{theorem}

\begin{proof}
Define the vector $X$ such that $x_0 = -2^W$ and for $i \in [k]$: $x_i = p_i$. Then every sequence $s$ that zeroes $P$ also zeroes $X$ because the excess of $2^W$ weight from indices $1$ to $k$ is transacted to index $0$. Each transaction $\TRANSACT{i}{2^\ell}{j}$ has a size that is a power of two, so by Lemma~\ref{lemma_signed_bits_cancellation} it can decrease the total number of non-zero signed-bits in the representation of $X$ by at most $2$, one in the representation of $x_i$ and one in the representation of $x_j$. We need to zero $|\phi(X)| = |\phi(P)| + |\phi(-2^W)| = |\phi(P)|+1$ signed-bits, so any sequence that zeroes $X$ must have at least $\frac{|\phi(P)|+1}{2}$ transactions. This proves that $\ceil{\frac{|\phi(P)|+1}{2}} \le \lambda(P)$.

let $j = argmax_{i \in [k]}{|\phi(x_i)|}$ be the ``anchor index''. We can zero its bits ``for free'' by taking care of all other indices as follows: for $i \ne j$, if $\phi(x_i)[\ell] = 1$ we apply the transaction $\TRANSACT{i}{2^\ell}{j}$, and if $\phi(x_i)[\ell] = \overline{1}$ we apply the transaction $\TRANSACT{j}{2^\ell}{i}$. These transactions zero every $x_i$ for $i \ne j$, and since $\sum_{i=0}^{k}{x_i} = 0$ is an invariant, we get that the sequence also zeroes $x_j$. Overall, this sequence requires $|\phi(P)|+1 - M(P)$ transactions, and this proves $\lambda(P) \le |\phi(P)|+1 - M(P)$.
\end{proof}

\begin{remark}
\label{remark_approximation}
\newtext{Observe that the upper and lower bounds are tight up to a factor of $2$. This means that we could derive a $2$-approximation algorithm by generating the sequence in the proof of the upper bound and convert it into an LPM TCAM: The anchor is the target of the match-all rule, and we construct the other rules one per transaction, in order of non-increasing size of the transactions, such that addresses are always "taken" from/to the anchor by the other targets.}
\end{remark}

\begin{remark}
\label{remark_improved_phi_upper_bound}
The upper bound in Theorem~\ref{theorem_bounds_signed_bits} can be improved further, by changing the anchor index throughout the levels. If $\phi(x_{j_0})$ is dense in non-zero signed-bits in the lower levels, say up to level $d_0$, then we can pick $j_0$ as the anchor and generate the transactions of sizes at most $2^{d_0}$ against it. Then, we can decide to change the anchor to $j_1$ for all the transactions up to size $2^{d_1}$ (for $d_1 > d_0)$, and so on. In every anchor-change, we might have to add another transaction due to carry that accumulates in the current anchor. Because of this carry, it is not necessarily beneficial to switch anchors too frequently. However, it is likely that by dividing $W$ into more than one bulk of levels we can improve the upper bound. That being said, the expression $|\phi(P)|+1-M(P)$ has a simple closed-form and is already a $2$-approximation.
\end{remark}

\begin{theorem}
\label{theorem_tight_signed_bits_bounds}
The bounds of Theorem~\ref{theorem_bounds_signed_bits} are tight. More formally, there exist partitions such that:
\begin{enumerate}[label=(\arabic*)]
    \item $\ceil{\frac{|\phi(P)|+1}{2}} < \lambda(P) = |\phi(P)|+1 - M(P)$ (tight UB).
    \item $\ceil{\frac{|\phi(P)|+1}{2}} = \lambda(P) < |\phi(P)|+1 - M(P)$ (tight LB).
    \item $\ceil{\frac{|\phi(P)|+1}{2}} = \lambda(P) = |\phi(P)|+1 - M(P)$ (sandwiched). In particular, for $k=2$ this is always the case.
\end{enumerate}
\end{theorem}

\begin{proof}
The upper bound is tight for any partition in which each transaction except the last can zero at most one signed-bit. For example, any partition where the signed-bits representation of every weight does not have any coefficient that is $\overline{1}$, e.g. $P = [5,5,5,1]$ (with $W=4$) has $\lambda(P) = 6$ by $\TRANSACT{4}{1}{3}\TRANSACT{2}{1}{1}\TRANSACT{3}{2}{1}\TRANSACT{3}{4}{2}\TRANSACT{2}{8}{1}\TRANSACT{1}{16}{0}$. Since $\phi(P) = [101,101,101,1]$, we have $|\phi(P)| = 7$, $M(P) = 2$. Here the lower bound is not tight.

The lower bound is tight for any partition in which each transaction zeroes two signed-bits, e.g. $P = [1,3,12]$ for which $\lambda(P) = 3$ by $\TRANSACT{1}{1}{2}\TRANSACT{2}{4}{3}\TRANSACT{3}{16}{0}$. Since $\phi(P) = [1, 10\overline{1},10\overline{1}00]$, we have $|\phi(P)| = 5$, $M(P) = 2$. Here the upper bound is not tight.

The partitions in which the bounds are equal are those where a single weight participates in all the transactions, the anchor, and each transaction zeroes two signed-bits. For example, $P = [15,4,45]$ (with $W=6$): $\lambda(P) = 4$ by $\TRANSACT{3}{1}{1}\TRANSACT{2}{4}{3}\TRANSACT{1}{16}{3}\TRANSACT{3}{64}{0}$. Since $\phi(P) = [1000\overline{1},100,10\overline{1}0\overline{1}01]$, we have $|\phi(P)| = 7$, $M(P) = 4$. In particular, when $k=2$ this is always the case (originally shown by \cite{AccurateExp}): Denote $m_i = |\phi(p_i)|$ and $m = \min(m_1,m_2)$, then the upper bound is $|\phi(P)|+1-M(P) = m+1$. Because $p_2 = 2^W - p_1$ and $|\phi(p_2)| = |\phi(-p_2)|$, by Lemma~\ref{lemma_signed_bits_cancellation}, $|m_1 - m_2| \le 1$, so no matter whether $m_1 = m_2 = m$ or $|m_1 - m_2| = 1$, we get that the lower bound also equals $\ceil{\frac{|\phi(P)|+1}{2}} = m + 1$.
\end{proof}

\newtext{The previous theorems dealt with LPM TCAMs. Clearly the upper bound holds also for general TCAMs.} Theorem~\ref{theorem_non_lpm_lower_bound_signed_bits} below gives a lower bound on the minimal size of a general TCAM.

\begin{lemma}
\label{lemma_intersection_always_power_of_two}
\newtext{Let $r_1,\ldots,r_m$ be $m$ patterns of rules (over $\{0,1,*\}^W$).
The number of addresses that satisfy all $m$ rules 
 is a power of $2$, or $0$.}
\end{lemma}

\begin{proof}
\newtext{Consider first the case $m = 2$. If there is a bit in which $r_1$ is $0$ and $r_2$ is $1$ or vice versa then there are no addresses that satisfy both $r_1$ and $r_2$. Otherwise, an address satisfies both $r_1$ and $r_2$ if and only if it also satisfies a single pattern that (a) has a don't-care at position $i$ if both $r_1$ and $r_2$ have don't-cares at position $i$; (b) it has $0$ at position $i$ if either one of $r_1$ and $r_2$ has a $0$ at position $i$ and $1$ if either one of $r_1$ and $r_2$ has a $1$ at position $i$.}

\newtext{If $m>2$, we can apply the argument in the previous paragraph to replace $r_1$ and $r_2$ by a single pattern without changing the set of addresses that satisfy all rules. We repeat this argument until either we identify that the intersection is empty or we end up with a single pattern. The size of the intersection is then $2^h$ where $h$ is the number of don't-care bits in the final pattern.}
\end{proof}

\begin{theorem}
\label{theorem_non_lpm_lower_bound_signed_bits}
Let $P$ be a partition, and without loss of generality assume that $|\phi(p_1)| \ge |\phi(p_2)| \ge \ldots \ge |\phi(p_k)|$. Then $n(P) \ge \max_{i=1,\ldots,k}{\lg ( |\phi(p_i)| + 1) + i - 1}$.
\end{theorem}

\begin{proof}
\newtext{Consider a set of TCAM rules for $P$.
Let $r_j$ be the rule of target $j$ with least priority. Denote the number of rules below (lower priority) $r_j$ by $x_j$ and the the number of rules above $r_j$, including $r_j$, by $m_j$. Let $f(j) = \lg (|\phi(p_j)| + 1)$. We  prove that $m_j \ge f(j)$ and this implies the lemma.
Indeed, the least priority rule of each target induce a permutation $\pi$ on the targets. That is $\pi(j)$ is the number of targets whose least priority rule is below (or equal) the least priority rule of target $j$.
It follows that the size of the TCAM is at least $ \max_{j=1,\ldots,k}{(f(j) + \pi(j) - 1)}$. This expression is minimized by the identity permutation because $f$ is non-increasing, so $\pi$ should be non-decreasing to minimize this.} 

\newtext{Now we prove $m_j \ge f(j)$. Fix a target $j$ and look at the $m_j$ rules above and including its lowest-priority rule. These $m_j$ rules define up to $2^{m_j} - 1$ non-empty intersections of subsets of rules, that include at least one rule of target $j$.
An intersection of each such  subset $S$ is identified with the target that belongs to the rule of highest priority in $S$.
Since
these $m_j$ rules define the  $p_j$ addresses that are allocated to target $j$ we must get $p_j$ by the following process.
First, go over all the intersections of a single rule, and add the number of addresses that they cover if their target is $j$. Then, go over intersections of size $2$, and subtract their size if they represent over-counting of addresses: either because they are the intersection of two rules with target $j$, or because this intersection has targets $i,j$ for $i \ne j$ and is associated with $i$ (higher priority). Then go over intersections of size $3$, and add their sizes in case that we over-subtracted them, and so on. In general, this process is done according to the inclusion-exclusion principle.}

\newtext{By Lemma~\ref{lemma_intersection_always_power_of_two} each addition or subtraction in this process is of a  power of two, so the final number that we get, which is $p_j$, has at most $2^{m_j} - 1$ signed-bits in its signed-bits representation. This is an upper bound on $|\phi(p_j)|$ because as noted some intersections may not correspond to additions or subtractions, and also there may be cancellations (add $2^h$ and subtract $2^h$) or carry (e.g.\ adding $2^h$ twice is equivalent to adding $2^{h+1}$ once). Therefore, we conclude that $|\phi(p_j)| \le 2^{m_j} - 1$. Extracting $m_j$ we get: $\lg (|\phi(p_j)|+1) \le m_j$.}
\end{proof}

We note that the lower bound of Theorem~\ref{theorem_non_lpm_lower_bound_signed_bits} is likely very loose for many partitions, \newtext{because of the strong assumption that all possible intersections are non-empty may be false, as well as the ``wishful'' scenario such that every signed-bit of $p_i$ will correspond to a unique intersection, etc.} However, this lower bound enables us to determine hard-partitions even for general TCAMs, despite the fact that no feasible algorithm to compute or approximate such TCAMs exist.

\begin{example}
\label{example_non_prefix_rules}
Consider the partition $P = [683,341]$ from Remark~\ref{remark_non_prefix_remark}. In signed-bits, we have $p_1 = 10\overline{1}0\overline{1}0\overline{1}0\overline{1}0\overline{1}$ and $p_2 = 101010101$. Then $|\phi(p_1)| = 6$ and $|\phi(p_2)| = 5$, and by Theorem~\ref{theorem_non_lpm_lower_bound_signed_bits} we get $n(P) \ge \max(\lg(7),\lg(6)+1) \Rightarrow n(P) \ge 4$. This means that every TCAM that represents $P$ must have at least $4$ rules. \newtext{In Remark~\ref{remark_non_prefix_remark} we showed that $n(P) \le 5$.}
\end{example}

\begin{theorem}
\label{theorem_tight_signed_bits_lower_bound_general_tcam}
\newtext{The lower bound of Theorem~\ref{theorem_non_lpm_lower_bound_signed_bits} cannot be made tighter (larger) in general. That is, there are partitions such that $n(P) = \max_{i=1,\ldots,k}{\lg ( |\phi(p_i)| + 1) + i - 1}$.}
\end{theorem}

\begin{proof}
\newtext{As a simple example, consider $P = [21,11]$, which has $k=2$, $W=5$, and its signed-bits representation is $[10101,10\overline{1}0\overline{1}]$. By Theorem~\ref{theorem_non_lpm_lower_bound_signed_bits}, $n(P) \ge 3$. On the other hand, $n(P) \le 3$ by the following set of rules: $ \{{*}{*}{0}{0}{0} \to 2 , {0}{0}{*}{*}{*} \to 2 , {*}{*}{*}{*}{*} \to 1 \}$.}
\end{proof}

\begin{corollary}
\label{corollary_lower_bound_non_lpm_size}
There exists a partition $P$ of $2^W$ to $k$ parts that satisfies $n(P) \ge \lg (W - \ceil{\lg k} + 3) + k-2$.
\end{corollary}

\begin{proof}
Let $h = W - \ceil{\lg k}$, and define an initial partition $Q$ of $2^W$ such that for every $i$, $q_i$ is either $2^h$ or $2^{h+1}$, and $\forall i: q_k \ge q_i$. Next define $\Delta = \sum_{j=1}^{\floor{h/2}}{2^{h-2j}}$ and perturb $Q$ to get $P$ as follows. If $k$ is even then $p_i = q_i + {(-1)}^i \Delta$. If $k$ is odd, then $q_k = 2^{h+1}$, we define $p_i = q_i + {(-1)}^i \Delta$ for $i \le k-2$, $p_{k-1} = q_{k-1} - \Delta$ and $p_k = q_k + 2\Delta$. One can verify that $\sum_{i=1}^{k}{p_i} = 2^W$. Observe that $|\phi(2^{h+1} \pm 2\Delta)| = |\phi(2^h \pm \Delta)| = |\phi(2^{h+1} \pm \Delta)| = \floor{\frac{h}{2}} + 1$. Thus $|\phi(p_i)| = \floor{\frac{h}{2}} + 1$ for all $i$, and by Theorem~\ref{theorem_non_lpm_lower_bound_signed_bits}: $n(P) \ge \lg (\floor{\frac{h}{2}} + 2) + k-1 \ge \lg (W - \ceil{\lg k} + 3) + k-2$.
\end{proof}

\begin{example}
Let $W = 7$ and $k = 5$. Then $h = 4$, $\Delta = 2^2 + 2^0 = 5$, and $Q = [16,16,32,32,32]$. Since $k$ is odd, we get $P = [11,21,27,27,42] = [10\overline{1}0\overline{1},10101,100\overline{1}0\overline{1},100\overline{1}0\overline{1},101010]$. $\forall i: \phi(p_i) = 3$, and $n(P) \ge 6$.
\end{example}

\section{Average-Case Analysis of LPM TCAM Size}
\label{section_average_case}
In this section we prove the following two theorems regarding the expected complexity (revisit Definition~\ref{definition_avg_length_of_partition}):

\begin{theorem}
\label{theorem_avg_case_rand_bits_rough}
$L(k) \in [\frac{1}{6},\frac{1}{5}]$ and $L(2) = \frac{1}{6}$.
\end{theorem}

\begin{definition}[Random Walk]
\label{definition_random_walk}
We denote by $\randwalk{p}{n}$ the expected distance reached by a symmetric random walk with $n$ independent steps, each moves left with probability $p$, right with probability $p$, and does not move with probability $1-2p$.
\newtext{Explicitly, $RW(p,n) = \sum_{0 \le \ell+r \le n}{\binom{n}{\ell,r} \cdot P(\ell+r) \cdot |\ell-r|}$ where $P(c) = p^{c} \cdot (1-2p)^{n-c}$.}
\end{definition}

\begin{theorem}
\label{theorem_bound_per_k}
$L(k) \le \frac{1}{6} + c(k)$ where $c(k) \equiv \frac{1 + \randwalk{\frac{1}{6}}{k-1}}{2k}$.
Furthermore, $\lim_{k \to \infty}{\sqrt{6\pi k} \cdot c(k)} = 1$, 
$\lim_{k \to \infty}{L(k)} = \frac{1}{6}$.
\end{theorem}

\newtext{Note that Theorem~\ref{theorem_avg_case_rand_bits_rough} is stronger for small values of $k$, while Theorem~\ref{theorem_bound_per_k} slowly improves when $k$ grows such that in the limit ($k \to \infty$) it is tight. Also, when $k$ is small, we can compute $c(k)$ directly, without approximating it as $\frac{1}{\sqrt{6\pi k}}$.}

Our analysis of $L(k)$, which is an asymptotic measure for $\expect{\lambda(P)}$ when $W \to \infty$, consists of two main parts:

    \emph{(i)} We show that sampling $P$ uniformly is asymptotically equivalent to sampling it uniformly from a smaller set, referred to as an \emph{$m$-nice} set (Definition~\ref{definition_good_partitions}) such that the lower $m = W - o(W)$ bits of every weight $p_i$ give a uniformly random value in $[0,2^m-1]$, and that we can focus only on the contribution of transactions of these levels (Lemmas~\ref{lemma_approx_bins_by_rand_bits}-\ref{lemma_nice_set_and_rand_bits}).
    
    \emph{(ii)} We formulate the problem as a (single player) game played by each weight $p_i$ over its least significant $m$ bits, such that a transaction correspond to a turn in this game (Definitions~\ref{definition_game_zeroing_bits}-\ref{definition_game_strategies}). We bound $\lambda(P)$ by analyzing the optimum strategy for the game (lower bound) and the strategies that are induced by \emph{Random Matcher (RM)} and \emph{Signed Matcher (SM)} in Algorithm~\ref{alg_match_variants} (upper bounds). Theorem~\ref{theorem_avg_case_rand_bits_rough} analyzes RM, assisted by Definition~\ref{definition_game_zeroing_bits_random_variables} and Lemmas~\ref{lemma_zeroing_strategies_expectations}-\ref{lemma_asymptotic_rough}.
    Theorem~\ref{theorem_bound_per_k} analyzes SM, assisted by Definitions~\ref{definition_f_g_counters}-\ref{definition_signed_probabilities} and Lemmas~\ref{signed_bits_prob_strong}-\ref{lemma_random_walk_with_skip_convert_to_fewer_steps_no_skip}.

\newtext{Once we show step (i), it becomes almost trivial to prove that $\frac{1}{6} \le L(k) \le \frac{1}{3}$ given Theorem~\ref{theorem_bounds_signed_bits}. Indeed, it is a known property~\cite{SignedBitsFrequences} that a uniformly random integer $q \in [0,2^m-1]$ satisfies $\lim_{m \to \infty}{\frac{|\phi(q)|}{m}} = \frac{1}{3}$. Therefore we can substitute $\expect{|\phi(P)|} = \frac{k \cdot W}{3} + o(kW)$ in Theorem~\ref{theorem_bounds_signed_bits}. That being said, we prove the lower bound more rigorously as part of Theorem~\ref{theorem_avg_case_rand_bits_rough}.}

\begin{definition}[$\boldsymbol{m}$\textbf{-nice} partitions]
\label{definition_good_partitions}
Let $A$ be a collection of ordered-partitions. Sample $P \in A$ uniformly, and define $k$ random variables $y_i \equiv p_i \pmod{2^m}$ for $i \in [k]$. We say that $A$ is $\boldsymbol{m}$\textbf{-nice} if every $y_i$ is uniformly distributed in $[0,2^m-1]$,\footnote{Which means that in the binary representation of $y_i$, every bit is uniform.}
and the $y_i$'s are $(k-1)$-wise independent.
\end{definition}

The following notations are used in the rest of the  paper.

\begin{definition}[\textbf{Notations}]
\label{definition_notations_average_case}
(1) Denote $m = W - \floor{\sqrt{W}}$.
(2) Denote by $\lambda_m(P)$ the number of transactions that are produced by the Bit Matcher algorithm for the partition $P$ until the least significant $m$ bits of every weight are zero.
(3) Let $S(k,W)$ be the set of all ordered partitions of $2^W$ to $k$ positive parts.
(4) Let $A \subset S(k,W)$ be the $m$-nice set constructed in Lemma~\ref{lemma_approx_bins_by_rand_bits} below.
(5) $y \sim Y$ samples $y$ uniformly from the set $Y$.
\end{definition}

The following lemma constructs $m$-nice sets.

\begin{lemma}
\label{lemma_approx_bins_by_rand_bits}
There exists an $m$-nice subset of $S(k,W)$ whose size is at least $(1 - k^{2-d}) \cdot |S(k,W)|$, where $d = \frac{W - m}{\lg k}$. (The claim is meaningful only if $d > 2$.)
\end{lemma}

\begin{proof}
We prove the lemma by constructing $A$.
We think of generating a partition by throwing $2^W$ balls into $k$ bins such that no bin is empty. Let $A$ be the subset of $S(k,W)$ of partitions that are generated in the following way: First, partition the balls to $2^{W-m}$ bundles each containing $2^m$ balls. Then partition the bundles to the $k$ bins such that no bin is empty, i.e.\ we select an ordered-partition $P'$ of $2^{W-m}$ for the bundles. Then, for each bin $i$, $1\le i\le k-1$, choose some integer $n_i \in [0,2^m-1]$, and move $n_i$ balls from bin $i$ to the last bin. Because $n_i < 2^m$ all the bins remain non-empty as we require. Each partition in $A$ has a unique representation as a tuple $(P',n_1,\ldots,n_{k-1})$, and no two tuples give the same partition (the mapping is reversible).

\newtext{One can verify that $\rho(A) \equiv \frac{|A|}{S(k,W)} > 1 - k^{2-d}$, indeed:
$$\rho(A) = \frac{2^{m\cdot(k-1)} \cdot \binom{2^{W-m}-1}{k-1}}{\binom{2^W-1}{k-1}}
= \prod_{i=0}^{k-2} {\frac{2^m \cdot (2^{W-m}-1-i)}{2^W-1-i}}$$
 $$= \prod_{i=0}^{k-2} {\frac{2^W-(i+1)\cdot 2^m}{2^W-(i+1)}}
  = \prod_{i=0}^{k-2} {\left(1 - \frac{(i+1)\cdot (2^m-1)}{2^W-(i+1)}\right)}$$
$$> {\left(1 - \frac{(k-1)\cdot (2^m-1)}{2^W} \right)}^{k-1}
> 1 - \frac{{(k-1)}^2\cdot (2^m-1)}{2^W}$$}

\newtext{The final inequality used the fact that ${(1-x)^n} > 1 - nx$ for $x \ne 0$. By substituting $m = W - d \cdot \lg k$, we get that the complement is:
$
\rho(A^c)
< \frac{{(k-1)}^2\cdot (2^{W - d \cdot \lg k}-1)}{2^W}
< {(k-1)}^2\cdot k^{-d}
< k^{2-d}
$.}

Also note that the random variables $y_i$ ($1\le i\le k$) as in Definition
\ref{definition_good_partitions} are distributed uniformly in $[0,2^m-1]$. The uniformity of $y_k$ and the $(k-1)$-wise independence of the $y_i$'s follow from the fact that 
$\sum_{i \in [k]}{p_i} \equiv 0\ (mod\ 2^m)$. It follows that $A$ is $m$-nice.
\end{proof}

Next, we prove that sampling a partition from $A$, or sampling it from $S(k,W)$, is ``effectively equivalent''.
\begin{lemma}
\label{lemma_nice_set_and_rand_bits}
Consider any of the algorithms $BM$, $RM$ or $SM$ (see Algorithm~\ref{alg_match_variants}). Denote by $N(P)$ the number of transactions that the particular algorithm generates for input $P$, and by $N_m(P)$ the number of transactions that the algorithm generates in levels $\ell < m$. Then:
$$\lim_{W \to \infty}{\frac{\expectExplicit{N(P)}{P \sim S(k,W)}}{W}}
= \lim_{W \to \infty}{\frac{\expectExplicit{N_m(P)}{P \sim A}}{W}}$$
\end{lemma}

\begin{proof}
\newtext{The high-level idea is that because $\frac{|A|}{|S(k,W)|} \to 1$, and because $\frac{m}{W} \to 1$, partitions which are not in $A$ or transactions that do not affect the $m$ least significant levels are asymptotically negligible. The rest of the proof is technical.}

\newtext{First, note that because no $i \in [k]$ participates in more than one transaction per level, there can be at most $k$ transactions per level, which is why for every partition $P$: $N(P) \le Wk$ and $0 \le N(P) - N_m(P) \le (W-m)k$. This inequality remains true in expectation when $P \sim A$ so:}
\begin{equation}
\label{equation_lambda_m_bounds}
\expect{N_m(P)} \le \expect{N(P)} \le \expect{N_m(P)} + (W-m)k
\end{equation}

\newtext{Denote $A^c = S(k,W) \setminus A$. By the law of total expectation:
$
\expectExplicit{N(P)}{P \sim S(k,W)} =
\expectExplicit{N(P)}{P \sim A} \cdot \Pr[P \in A] +
\expectExplicit{N(P)}{P \sim A^c} \cdot \Pr[P \in A^c] $}

\newtext{By Lemma~\ref{lemma_approx_bins_by_rand_bits} the subset $A$ satisfies $\Pr[P \in A^c] \le k^{2-(\sqrt{W}/ \lg k)}$. Overall, the contribution of this part to the expectation is non-negative, and bounded by $Wk \cdot k^{2-(\sqrt{W}/ \lg k)} = o(\frac{1}{W})$, Therefore by the sandwich theorem we can conclude that:}
\begin{equation}
\label{equation_expectation_like_nice_set}
\newtext{
\lim_{W \to \infty}{\frac{\expectExplicit{N(P)}{P \sim S(k,W)}}{W}}
= \lim_{W \to \infty}{\frac{\expectExplicit{N(P)}{P \sim A}}{W}}
}
\end{equation}

\newtext{Since $k$ is fixed and $W-m = o(W) = o(m)$, we get by another application of the sandwich theorem on Equation~(\ref{equation_lambda_m_bounds}), that:}
\begin{equation}
\label{equation_expectation_sandwiched}
\newtext{\lim_{W \to \infty}{\frac{\expectExplicit{N(P)}{P \sim A}}{W}}
= \lim_{W \to \infty}{\frac{\expectExplicit{N_m(P)}{P \sim A}}{W}}}
\end{equation}

The claim follows by combining Equations~(\ref{equation_expectation_like_nice_set}) and (\ref{equation_expectation_sandwiched}).
\end{proof}

When BM is considered in Lemma~\ref{lemma_nice_set_and_rand_bits}, $N(P) = \lambda(P)$ and $N_m(P) = \lambda_m(P)$. Next, we rephrase this problem as a game.

\begin{definition}[Zeroing Bits Game]
\label{definition_game_zeroing_bits}
Let $v$ be an infinite string of bits, which we think of as a generalization of a number such that the bit in location $\ell$ represents a value of $2^\ell$. We denote $v \pmod {2^m}$ by $\hat{v}$. In turn $t$ of the game, find the lowest bit that is $1$, denote its location by $d_t$. The player should either add $2^{d_t}$ to $v$ or subtracts $2^{d_t}$ from $v$. Repeat until $\hat{v} = 0$. (The number of turns can be thought of as the player's utility.)
\end{definition}

By adding or subtracting $2^{d_t}$ we are guaranteed that bit $d_t$ is zeroed. If we added $2^{d_t}$, we earn more zeroed bits if there were consecutive $1$s at levels following $d_t$ that became $0$s due to carry, and if we subtracted $2^{d_t}$ then we ``earn'' more zeroed bits if at the levels following $d_t$ there are consecutive $0$s.

\begin{remark}
\label{remark_k_games}
A running of $BM$, $RM$ or $SM$ on a partition $P$ induces a strategy for $k$ simultaneous games as follows:
\begin{enumerate}
    \item The string $v$ of the $i$th game satisfies $v \equiv p_i (mod\ 2^m)$. Higher bits in $v$ are picked uniformly at random.
    
    \item The choices the algorithm makes are adding $2^d$ when $i$ participates in a transaction
    $\TRANSACT{j_1}{2^d}{i}$ and subtracting $2^d$ when $i$ participates in a transaction
    $\TRANSACT{i}{2^d}{j_2}$.
\end{enumerate}
\end{remark}

\begin{definition}[Strategies $OPT$, $RND$, $MIX$]
\label{definition_game_strategies}
We define the following strategies for the game in Definition~\ref{definition_game_zeroing_bits}, for the action taken in turn $t$, at level $d_t$ (by definition: $v[d_t] = 1$):
\begin{enumerate}[leftmargin=*,label=(\arabic*),topsep=0pt]
    \item $OPT$: if $v[d_t+1] = 1$ add $2^{d_t}$, otherwise subtract $2^{d_t}$. (Equivalently, in signed-bits representation: add if $\phi(v)[d_t] = -1$, subtract if $\phi(v)[d_t] = 1$.)
    \item $RND$: Choose to add/subtract $2^{d_t}$ with probability half.
    \item $MIX$: At any turn, choose between $OPT$ and $RND$ with probability half, \textbf{independent} of the bits of $v$. 
\end{enumerate}
\end{definition}

\begin{definition}[Game Random Variables]
\label{definition_game_zeroing_bits_random_variables}
Revisit Definition~\ref{definition_game_zeroing_bits} of the game. We will consider settings in which $v$ is chosen randomly from some distribution. In such a setting we define the following random variables: $T$ is the number of turns played in the game (if $\hat{v} = 0$ initially, then $T=0$), and $D_1,D_2,\ldots,D_T$ are such that $D_t = d_{t+1} - d_t$. Note that for notation purposes, we consider $d_{T+1}$ to be the lowest bit that is $1$ when the game ends, i.e. the level of turn $T+1$, if there would have been another turn. Also, for $D_T$ to be well defined, we assume that $v$ contains infinitely many $0$s and infinitely many $1$s.\footnote{This assumption is stronger than what is needed for $D_T$ to be well defined, but our settings allow it.}
We refer to $D_t$ as the number of zeroed bits in turn $t$, and to $T$ as the number of turns.
\end{definition}

\begin{remark}
\label{remark_finite_bits_horizon}
\newtext{The game is defined on an infinite $v$ to make all the $D_t$ random variables \newtext{identically distributed} with respect to playing a game according to the strategies in Definition~\ref{definition_game_strategies}, to simplify Lemma~\ref{lemma_zeroing_strategies_expectations} by preventing truncation effects. Although $v$ is infinite, observe that each of these strategies requires at most $m+1$ bits, and moreover $T$ is completely determines by the first $m$ bits. We rely on this observation in Lemma~\ref{lemma_invert_expectation_by_renewal}.}
\end{remark}

The next lemma analyzes the gist of the above strategies.

\begin{lemma}
\label{lemma_zeroing_strategies_expectations}
Let the game of Definition~\ref{definition_game_zeroing_bits} be played on a value $v$ sampled such that each of its bits is uniform and independent of the others. Recall the strategies in Definition~\ref{definition_game_strategies}, then:
\begin{enumerate}
    \item Whenever a turn is played at bit $d$, any bit $v[\ell]$ for $\ell \ge d+1$ is uniformly random.
    \item For a fixed strategy $OPT$ or $RND$ or $MIX$, the random variables $D_t$ are identically distributed and independent.
    \item $\expect{D_t^{OPT}} = 3$, $\expect{D_t^{RND}} = 2$ and $\expect{D_t^{MIX}} = 2.5$.
\end{enumerate}
\end{lemma}

\begin{proof}
\newtext{We start from the first part, using a simple induction: Initially all the bits of $v$ are uniform, and the first turn begins at the least significant bit that is $1$, so every bit above level $d$ is uniform. Next, observe that if the action played is subtraction, then the bits above it remain unchanged. If the action is addition, the bits above it may change due to carry. However, the bits that change are only those who get zeroed, and the last bit which is $0$ and becomes $1$ that ``stops'' the carry. This is the bit that will be played at the next turn, and every bit above it is unchanged, and therefore still uniform.}

\newtext{Now that we have the uniformity of bits in our hand, the claim regarding expectations is simple. Any strategy zeroes at least one bit, $v[d]$. Other than that, by definition, $OPT$ zeroes $v[d+1]$ as well. $OPT$ earns more zeros as long as a streak of bits continues such that $v[d+i] = v[d+1]$ for $i \ge 2$, so the total number of zeroed bits for $OPT$ is $2+G$ where $G$ is a geometric random variable. The expectation that we get is $\expect{D_t^{OPT}} = 2 + \expect{G} = 3$. On the other hand, the opposite choice of $OPT$, by definition, yields only a single zeroed bit. Therefore, $\expect{D_t^{RND}} = \frac{1}{2}(3+1) = 2$ and also $\expect{D_t^{MIX}} = \frac{1}{2}(3+2) = 2.5$. The expectation for $MIX$ subtly relies on the fact that it chooses between $OPT$ and $RND$ in a way independent from $v$.\footnote{\newtext{To see that, consider a scheme in which if $v[d+1]=v[d+2]$ we play $OPT$, and otherwise $RND$. It is still probability half to choose either, but if $v[d+1]=v[d+2]$ then $OPT$ is chosen, $3$ bits are guaranteed to be zeroed with one more in expectation, and if $v[d+1] \ne v[d+2]$ then $RND$ will zero either $1$ or $2$ bits. The expectation is $\frac{1}{2}(4 + \frac{3}{2}) = \frac{11}{4}$, larger than $2.5$.}}}

\newtext{The discussion above also implies that the random variables $D_t$ are independent since they count disjoint sub-sequences of bits.}
\end{proof}

$OPT$ is not only better in expectation, in fact it is optimal.

\begin{lemma}
\label{lemma_zeroing_strategies_opt_is_opt}
Let the game of Definition~\ref{definition_game_zeroing_bits} be played on a value $v$. The strategy $OPT$ minimizes the number of turns.
\end{lemma}

\begin{proof}
Let $v$ be the string of bits. Let $S$ be some strategy, and denote the number of turns it plays on $v$ by $n^S$.
\newtext{As a first step, we present a strategy $S'$ such that $n^{S'} \le n^{S}$ and the last action of $S'$ is ``in sync'' with $OPT$: either both are addition, or both are subtraction. If $S$ is in sync with $OPT$, $S' = S$. If the last turn of $OPT$ is at level $m-1$, this is the last turn of $OPT$ anyway, whether it adds or subtracts, so we flip its action and abuse notation to keep calling it $OPT$ (instead of adding a notation for $OPT'$). Otherwise, the last turn of $OPT$ it played at level $\ell$ for some $\ell < m-1$.
\begin{enumerate}
    \item If $OPT$ applied $+2^{\ell}$: it means that $\hat{v}[\ell'] = 1$ for all $\ell' > \ell$ when the game starts. $\hat{v}[\ell]$ may be $1$ originally, or change to $1$ during the game. Anyway, since the last action of $S$ is subtraction, it must play a turn in each of the levels $\ell+1,\ldots,m-1$. We define $S'$ to play like $S$ until the turn on level $\ell+1$. This turn exists because $\ell+1 \le m-1$. Then $S'$ plays addition and the game ends. $n^{S'} \le n^S$.
    \item If $OPT$ applied $-2^{\ell}$: it means that $\hat{v}[\ell'] = 0$ for all $\ell' > \ell$ when the game starts. $\hat{v}[\ell]$ may be $1$ originally, or change to $1$ during the game. Anyway, since the last action of $S$ is addition, it can at best induce some carry up to level at most $\ell+1$, and must play the rest of the levels one by one (since its last action is addition). We define $S'$ to play like $S$ until the turn on level $\ell+1$. This turn exists because $\ell+1 \le m-1$. Then $S'$ plays subtraction and the game ends. $n^{S'} \le n^S$.
\end{enumerate}
}

Next, let $d_{i_1},\ldots,d_{i_t}$ be the levels where $S'$ adds a power of $2$ and $d_{j_1},\ldots,d_{j_s}$ the levels where $S'$ subtracts a power of $2$. Define $a^{S'} = \sum_{i=i_1}^{i_t}{2^{d_i}}$ and $b^{S'} = \sum_{j=j_1}^{j_s}{2^{d_j}}$. By definition $v \equiv b^{S'} - a^{S'} \pmod {2^m}$.

\newtext{The same can be defined for $OPT$, so $v \equiv b^{OPT} - a^{OPT} \pmod {2^m}$. Since we made sure that $S'$ is in sync with $OPT$, and because $0 \le a^{S'},b^{S'},a^{OPT},b^{OPT} < 2^m$, we get that $b^{S'} - a^{S'} = b^{OPT} - a^{OPT}$. This equality is actual (not modulo), denote this number by $r$ (possibly $r<0$). So we have two signed-bits representation of the same number. Recall that by Definition~\ref{definition_game_strategies} $OPT$ can be interpreted as playing by the canonical signed-bits representation. Therefore by Property~\ref{property_signed_bits_most_sparse}, we get that: $n^{OPT} = \norm{a^{OPT}} + \norm{b^{OPT}} = |\phi(r)| \le \norm{a^{S'}} + \norm{b^{S'}} = n^{S'} \le n^{S}$.}
\end{proof}

The following lemma converts us from talking about zeroed-bits in the game to the number of turns in a game.

\begin{lemma}
\label{lemma_invert_expectation_by_renewal}
Let $S$ be one of the strategies in Definition~\ref{definition_game_strategies}. Let $\hat{v}$ be an $m$-bits number, and augment it with additional uniform and independent bits to get $v$. Let $\mu$ be the expected number of bits that are zeroed by the first turn of $S$ on $v$. Then $\lim_{m \to \infty}{\frac{\expectExplicit{T}{\hat{v} \sim [0,2^m-1]}}{m}} = \frac{1}{\mu}$.
\end{lemma}

\begin{proof}
This claim follows directly from the \emph{elementary renewal theorem} \cite[Section 1.2]{RenewalTheorem} of a \emph{general renewal process}: Denote by $D_0$ the number of least significant bits that are initially zero in $\hat{v}$. The game on $v$ is played until $\sum_{t=0}^{t=T}{D_t} \ge m$. Since all the random variables $D_t$ for $t \ge 1$ are identically distributed by Lemma~\ref{lemma_zeroing_strategies_expectations}: $\lim_{m \to \infty}{\frac{\expectExplicit{T}{\hat{v} \sim [0,2^m-1]}}{m}} = \frac{1}{\mu}$.
\end{proof}

\begin{lemma}
\label{lemma_asymptotic_rough}
$\lim_{m \to \infty}{\frac{\expectExplicit{\lambda_m(P)}{P \sim A}}{m}} \in [\frac{k}{6},\frac{k}{5}]$. In particular, for $k=2$, the limit equals to $\frac{k}{6}$.
\end{lemma}

\begin{proof}
Let $P$ be a partition. According to Remark~\ref{remark_k_games}, both $BM$ and $RM$ induce a strategy for $k$ simultaneous games with input $v_i$ such that $\hat{v}_i \equiv p_i\ (mod\ 2^m)$ for the $i$th game. When $P \sim A$, we get $\hat{v}_i \sim [0,2^m-1]$ (and each of the bits of $\hat{v}$ is uniform). We emphasize that while these values are not totally independent, by Lemma~\ref{lemma_approx_bins_by_rand_bits} each is uniform, and they are $(k-1)$-wise independent such that the constraint on all of them is $\sum_{i \in [k]}{v_i} \equiv 0\ (mod\ 2^m)$.

Denote by $n^{BM}_i$ the number of turns in game $i$ when played by the strategy induced by $BM$, denote by $n^{RM}_i$ the number of turns in game $i$ when played by the strategy induced by $RM$, and by $n^{OPT}_i$ the minimum number of turns of a strategy for game $i$. By Lemma~\ref{lemma_zeroing_strategies_opt_is_opt}, $n^{OPT}_i \le n^{BM}_i$. Note that 
since every transaction corresponds to a turn in two games we have that $\sum_{i \in [k]}{n^{BM}_i} = 2 \lambda_m(P)$. Together with the linearity of expectation we get:
$$
\lim_{m \to \infty}{\frac{\expectExplicit{\lambda_m(P)}{P \sim A}}{m}}
\ge \frac{k}{2} \lim_{m \to \infty}{\frac{\expectExplicit{n^{OPT}_i}{\hat{v}_i \sim [0,2^m-1]}}{m}}
= \frac{k}{6}
$$
where the last equality follows from Lemma~\ref{lemma_invert_expectation_by_renewal}, substituting $\mu = 3$ for $OPT$ by Lemma~\ref{lemma_zeroing_strategies_expectations}.

When $k=2$, $n^{OPT}_i = n^{BM}_i$ for $i=1,2$. The reason is that if $p_1[d] = p_2[d] = 1$, since $p_1 + p_2 = 2^W$, we must have that $p_1[d+1] \ne p_2[d+1]$, therefore $OPT_1$ and $OPT_2$ choose opposite actions in the different games, which correspond exactly to the transaction computed by $BM$. For this reason, for $k=2$ the limit equals $\frac{k}{6}$.

For the upper bound, we proceed with $k \ge 3$. Note that since $BM$ computes a \emph{shortest} sequence for $P$ and $RM$ computes \emph{some} sequence, then $\lambda_m(P) \le \lambda(P) \le |RM(P)| \le |RM_{m}(P)| + \frac{(W-m)k}{2}$ where $|RM(P)|$ is the number of transactions of $RM$ and $|RM_m(P)|$ is the number of transactions of $RM$ up to level $m$.
We get: $\sum_{i \in [k]}{n^{BM}_i} \le \sum_{i \in [k]}{n^{RM}_i} + \frac{(W-m)k}{2}$.

Observe that per Remark~\ref{remark_k_games}, $RM$ induces the strategy $MIX$ defined in Definition~\ref{definition_game_strategies}, in each game. To clarify, a pair $(i,j)$
that participate in a transaction of $RM$, play both $OPT$ in the induced strategy if $b \equiv p_i[\ell+1] \oplus p_j[\ell+1] = 1$, and  play both $RND$ (anti-synchronized) if $b=0$.
$b$ is uniform by Lemma~\ref{lemma_zeroing_strategies_expectations} (note that $k \ge 3$ thus $p_i[\ell+1]$ and $p_j[\ell+1]$ are independent).
Since $\mu = 2.5$ for $MIX$ by Lemma~\ref{lemma_zeroing_strategies_expectations}, we get by Lemma~\ref{lemma_invert_expectation_by_renewal} and the linearity of expectation that
$
\lim_{m \to \infty}{\frac{\expectExplicit{\lambda_m(P)}{P \sim A}}{m}}
\le \frac{k}{2} \lim_{m \to \infty}{\frac{\expectExplicit{n^{RM}_i}{\hat{v}_i \sim [0,2^m-1]}}{m}}
= \frac{k}{5}
$
\end{proof}

\begin{proof}[Proof of Theorem~\ref{theorem_avg_case_rand_bits_rough}]
By Lemma~\ref{lemma_nice_set_and_rand_bits} and Lemma~\ref{lemma_asymptotic_rough}:
$L(k) = 
\lim_{W \to \infty}{\frac{\expectExplicit{\lambda(P)}{P \sim S(k,W)}}{kW}}
=
\lim_{W \to \infty}{\frac{\expectExplicit{\lambda_m(P)}{P \sim A}}{kW}}
= \lim_{m \to \infty}{\frac{m}{W} \cdot \frac{\expectExplicit{\lambda_m(P)}{P \sim A}}{km}} \in [\frac{1}{6},\frac{1}{5}]
$. 
Thus $\frac{1}{6} \le L(k) \le \frac{1}{5}$. When $k=2$, the limit equals $\frac{1}{6}$.
\end{proof}

Now that we are done with Theorem~\ref{theorem_avg_case_rand_bits_rough}, we proceed towards proving the alternative bound in Theorem~\ref{theorem_bound_per_k}.

\begin{definition}
\label{definition_f_g_counters}
Let $d \ge 1$. Let $x_1,\ldots,x_k$ be $k$ numbers such that $x_i \sim [0,2^d - 1]$ for $i < k$ and $\sum_{i=1}^{k}{x_i} \equiv 0\ ({\rm mod}\ 2^d)$.
We define the random variables $F_\ell \equiv \{ i \mid \phi(x_i)[\ell] = 1 \}$ and $G_\ell \equiv \{ i \mid \phi(x_i)[\ell] = -1 \}$. In words: $F_\ell$ is the set of indices with numbers that have $1$ in their $\ell$th signed-bit, and $G_\ell$ is similar for $-1$. We define $f_\ell \equiv |F_\ell|$ and $g_\ell \equiv |G_\ell|$.
\end{definition}

\begin{definition}
\label{definition_signed_probabilities}
Let $d \ge 1$ and let $n \sim [0,2^d-1]$.
We define $p_{\ell}^{a} \equiv \Pr\left[\phi(n)[\ell] = a \mid n \sim [0,2^d-1] \right]$ and $p_\ell \equiv \frac{p_{\ell}^{+1} + p_{\ell}^{-1}}{2}$.
\end{definition}

The next two lemmas analyze properties related to $f_\ell,g_\ell$.

\begin{lemma}
\label{signed_bits_prob_strong}
Let $d \ge 1$ and let $n \sim [0,2^d-1]$. Then:
\begin{enumerate}
    \item \label{claim_plus_minus_same} For every $\ell \le d-1$: $p_{\ell}^{+1} = p_{\ell}^{-1}$. For $\ell = d-1$, $p_{\ell}^{-1} = 0$.
    
    \item \label{converge_to_1over6} $p_0 = \frac{1}{4}$, and $(p_{\ell+1} - \frac{1}{6}) = \frac{1}{2}(\frac{1}{6} - p_{\ell})$.
\end{enumerate}
\end{lemma}

\begin{proof}
\newtext{We prove part~\ref{claim_plus_minus_same} first. Ignore $n=0$, so $1 \le n \le 2^d-1$. Since $n$ is positive, its leading non-zero signed-bit is $1$. If $\phi(n)[d-1] = -1$, then it must be that $\phi(n)[d] = 0$ and therefore $n > 2^d$ when we take into account the leading signed-bit that is $1$. Therefore it must be that $p_{d-1}^{-1} = 0$.}

\newtext{To prove that $p_{\ell}^{+1} = p_{\ell}^{-1}$ for $\ell < d-1$, we define the following one-to-one mapping between numbers $x \in [0,2^d-1]$ with $\phi(x)[\ell] = 1$ and $x' \in [0,2^d-1]$ with $\phi(x')[\ell] = -1$.
We map $x \in [0,2^d-1]$ with $\phi(x)[\ell] = 1$ to $x' \equiv x - 2^{\ell + 1} \pmod{2^d}$. Since $\ell < d-1$, $2^{\ell+1} < 2^d$, therefore $x$ does not pair with itself. Moreover, our mapping is a cyclic-shift, so it is injective. We claim that $\phi(x')[\ell]=-1$: If $x - 2^{\ell+1} \ge 0$ then the modulo did nothing, and $\phi(x')[\ell]=-1$ is trivial. Otherwise, $x - 2^{\ell+1} < 0$ and  its signed-bits representation  has a leading $-1$ in location $\ell$, and also $\phi(x')[d-1] = 0$. The modulo adds $2^d$ to $x-2^{\ell+1}$. Since level $d$ is separated by a $0$ at level $d-1$ from the leading non-zero signed-bit of $x-2^{\ell+1}$, this addition of $2^d$ is reflected in the signed-representation only by changing level $d$ to $1$. So we conclude that $x'[\ell] = -1$. To argue that the mapping is surjective too, note that
this mapping is invertible. The inverse maps in the same way any $x' \in [0,2^d-1]$ such that $\phi(x')[\ell] = -1$ to a unique $x \in [0,2^d-1]$ with $\phi(x)[\ell] = 1$.}

\newtext{We proceed to part~\ref{converge_to_1over6}. $p_0 = \frac{1}{4}$ since exactly half of the numbers are odd ($p_0^{+1} + p_0^{-1} = \frac{1}{2}$). Next, we compute $p_{\ell+1}$ in terms of $p_\ell$ (assuming that $\ell+1 \le d-1$). With probability $2p_\ell$ we have $\phi(n)[\ell] \ne 0$ which enforces $\phi(n)[\ell+1] = 0$. With probability $1 - 2p_\ell$ we have $\phi(n)[\ell] = 0$. In this case, we claim that $\Pr[\phi(n)[\ell+1] \ne 0 \mid \phi(n)[\ell] = 0] = \frac{1}{2}$. To be convinced, take a look at the automaton that converts binary form to signed-bits form in Figure~\ref{figure_signed_bits_automaton}. Consider our state right when we computed that $\phi(n)[\ell] = 0$, then it is either ``start'' or ``borrow''. Either has a probability of $\frac{1}{2}$ to generate another $0$, and probability $\frac{1}{2}$ to move to ``look-ahead'' which generates a non-zero signed-bit in level $\ell+1$. Then $2p_{\ell+1} = (1 - 2p_\ell) \cdot \frac{1}{2} \Rightarrow  p_{\ell+1} - \frac{1}{6} = \frac{3-6p_\ell - 2}{12} = \frac{1}{2}(\frac{1}{6} - p_\ell)$.}
\end{proof}

\begin{corollary}
\label{corollary_probabilities_in_levels}
\newtext{$p_\ell \to \frac{1}{6}$ for large $\ell$. The first few values are: $p_0 = \frac{1}{4}, p_1 = \frac{1}{8}, p_2 = \frac{3}{16}, p_4 = \frac{5}{32}, p_5 = \frac{11}{64}$. In particular, notice the oscillations around (over-and-under) $\frac{1}{6}$.}
\end{corollary}

\begin{lemma}
\label{upper_bound_per_level_with_random_walk}
Recall Definitions~\ref{definition_random_walk},\ref{definition_f_g_counters},\ref{definition_signed_probabilities}. If $\ell < d-1$ then $\expect{|f_\ell - g_\ell|} < 1 + \randwalk{p_\ell}{k-1}$, and if $\ell = d-1$ then $\expect{|f_\ell - g_\ell|} = \expect{f_\ell} = 2{p_\ell} \cdot k$.
\end{lemma}

\begin{proof}
For $\ell = d-1$, by Lemma~\ref{signed_bits_prob_strong} we know that $p_{\ell}^{-} = 0$ therefore $g_\ell = 0$. Also $p_{\ell}^{+} = 2p_\ell$, thus $\expect{|f_\ell - g_\ell|} = \expect{f_\ell} = 2{p_\ell} \cdot k$. For $\ell < d-1$, the expression $|f_\ell - g_\ell|$ is the absolute distance of a random walk with movement probabilities $p_{\ell}^{+} = p_{\ell}^{-} = p_{\ell}$ (by Lemma~\ref{signed_bits_prob_strong}), except that the $k$th step is not independent of the first $k-1$ steps, so we absorb an extra $1$ to our bound.
\end{proof}

The next two lemmas give useful limits for the analysis.

\begin{lemma}
\label{lemma_upper_bound_per_level_with_random_walk}
Recall Definitions~\ref{definition_random_walk},\ref{definition_signed_probabilities}. Then:\\
$\lim_{d \to \infty}{\frac{\sum_{\ell=0}^{d-1}{\randwalk{p_\ell}{n}}}{d}}
= \lim_{\ell \to \infty}{\randwalk{p_\ell}{n}}
= \randwalk{\frac{1}{6}}{n}$.
\end{lemma}

\begin{proof}
\newtext{We prove below the second equality. The first then follows as the average of a convergent sequence. We start by bounding the ratio $\frac{\randwalk{p \pm \epsilon}{n}}{\randwalk{p}{n}}$ as follows. First, note that we write an explicit expression for $\randwalk{p}{n}$ as:
$$
\randwalk{p}{n}
= \sum_{i=0}^{n}{
        \binom{n}{i} \cdot {(2p)}^i \cdot {(1-2p)}^{n-i} \cdot \randwalk{\frac{1}{2}}{i}
    }
$$
(This decides when we make a move. $RW(\frac{1}{2},i)$ then decides on the direction of each move.)
Denote $q = 1-2p$. We can bound our ratio by the maximum ratio over each summand separately, which cancels out some of the terms. We get:
$$\frac{\randwalk{p \pm \epsilon}{n}}{\randwalk{p}{n}} \le \max_{0 \le i \le n}{{\left(\frac{2(p \pm \epsilon)}{2p}\right)}^i \cdot {\left(\frac{1 - 2(p \pm \epsilon)}{1-2p}\right)}^{n-i}}
$$
$$
\le \max_{0 \le i \le n}{{\left(1 + \frac{\epsilon}{p}\right)}^i \cdot {\left(1 + \frac{2\epsilon}{q}\right)}^{n-i}}
\le {\left(1 + \epsilon \cdot \max\left(\frac{1}{p},\frac{2}{q}\right)\right)}^n
$$
For $p=\frac{1}{6}$ we get:
\begin{equation}
\label{eq_upper_ratio_bound}
\frac{\randwalk{\frac{1}{6} \pm \epsilon}{n}}{\randwalk{\frac{1}{6}}{n}} \le {(1 + 6\epsilon)}^n
\end{equation}
Similarly, we can find a lower bound by considering the minimum ratio (replace $\epsilon \to -\epsilon$ and $\max$ by $\min$):
\begin{equation}
\label{eq_lower_ratio_bound}
{\left(1 - 6\epsilon\right)}^n \le \frac{\randwalk{\frac{1}{6} \pm \epsilon}{n}}{\randwalk{\frac{1}{6}}{n}}
\end{equation}
}

\newtext{
We can write $p_\ell = \frac{1}{6} + \epsilon_\ell$, and by Lemma~\ref{signed_bits_prob_strong}, $|\epsilon_\ell| = |p_\ell - \frac{1}{6}| \le \frac{1}{2^{\ell}}$. Then by Equations~(\ref{eq_upper_ratio_bound})-(\ref{eq_lower_ratio_bound}), $\frac{\randwalk{p_\ell}{n}}{\randwalk{\frac{1}{6}}{n}} \in ((1-\frac{6}{2^\ell})^n,(1+\frac{6}{2^\ell})^n)$. Therefore $\lim_{\ell \to \infty}{\randwalk{p_\ell}{n}} = \randwalk{\frac{1}{6}}{n}$ (for $\ell < d$), and the sequence converges for $d \to \infty$.
}
\end{proof}

\begin{lemma}
\label{lemma_random_walk_with_skip_convert_to_fewer_steps_no_skip}
Recall Definition~\ref{definition_random_walk}: $\lim_{n \to \infty}{\frac{\randwalk{p}{n}}{\sqrt{4pn/\pi}}} = 1$.
\end{lemma}

\begin{proof}
\newtext{The fact that $\lim_{n \to \infty}{\frac{\randwalk{\frac{1}{2}}{n}}{\sqrt{2n/\pi}}} = 1$ is well-known~\cite{randomWalk}. Denote by $n' \le n$ the number of steps in which the random walk moves. $\expect{n'} = 2pn$ (recall that the probability to move is $2p$). Since each step is independent, by Chernoff's bound: $\Pr[|n' - 2pn| > \epsilon \cdot 2pn] \le 2 exp(-\frac{2}{3}pn \epsilon^2)$. Choose $\epsilon = (pn)^{-1/3}$, and denote for short $\alpha \equiv \Pr[|n' - 2pn| > 2(pn)^{2/3}] \le 2 exp(-\frac{2}{3}{(pn)}^{1/3})$ and $N \equiv 2pn + 2(pn)^{2/3}$. By the monotonicity of $\randwalk{\frac{1}{2}}{i}$ as a function of $i$ we have that $\randwalk{\frac{1}{2}}{n'} \le \randwalk{\frac{1}{2}}{N}$, and therefore by conditioning on the event whose probability is $\alpha$:
$$\randwalk{p}{n} \le \alpha \cdot n + (1-\alpha) \cdot \randwalk{\frac{1}{2}}{N}$$
Note that $\lim_{n \to \infty}{\alpha} = 0$ and that $\lim_{n \to \infty}{\frac{N}{2pn}} = 1$, thus:
$$
\lim_{n \to \infty}{\frac{\randwalk{p}{n}}{\sqrt{4pn/\pi}}}
\le 
\lim_{n \to \infty}{\frac{\randwalk{\frac{1}{2}}{N}}{\sqrt{4pn/\pi}}}
$$
$$
=
\lim_{n \to \infty}{\frac{\randwalk{\frac{1}{2}}{N}}{\sqrt{2N/\pi}} \cdot \sqrt{\frac{N}{2pn}}}
= 1
$$
Similarly, $\randwalk{p}{n} \ge (1-\alpha) \cdot \randwalk{\frac{1}{2}}{2pn - 2(pn)^{2/3}}$, and we get a lower bound of $1$ as well. Therefore by the sandwich theorem, $\lim_{n \to \infty}{\frac{\randwalk{p}{n}}{\sqrt{4pn/\pi}}} = 1$.}
\end{proof}

\begin{proof}[Proof of Theorem~\ref{theorem_bound_per_k}]
The equality $\lim_{k \to \infty}{L(k)} = \frac{1}{6}$ is due to the sandwich theorem, following the main claim of Theorem~\ref{theorem_bound_per_k}, together with the lower bound of $\frac{1}{6}$ from Theorem~\ref{theorem_avg_case_rand_bits_rough}.

To prove the main claim, we analyze the length of a sequence generated by $SM$ (see Algorithm~\ref{alg_match_variants}), which by definition satisfies $\lambda(P) \le |SM(P)|$. We derive $c(k)$ from the expected length of this sequence. We denote by $n^{SM}(P) \equiv |SM(P)|$, and similarly by $n^{SM}_m(P)$ the number of transactions at levels $\ell < m$.

In terms of our game, $SM$ induces a strategy for $k$ simultaneous games, such that it plays $OPT$ in every game. However, the number of transactions in each level is not half of the number of turns played at that level, it is $\max(|F_\ell|,|G_\ell|)$ because $SM$ saves turns compared to transactions only due to pairings. Also, by Lemma~\ref{lemma_nice_set_and_rand_bits}:
$\lim_{W \to \infty}{\frac{\expectExplicit{n^{SM}(P)}{P \sim S(k,W)}}{W}} = \lim_{W \to \infty}{\frac{\expectExplicit{n^{SM}_m(P)}{P \sim A}}{W}}$.

Denote by $n^{OPT}_i$ the number of turns to zero the least $m$ bits in game $i$ using $OPT$. Since $n^{SM}_m(P) = \sum_{\ell=0}^{m-1}{\max(f_\ell,g_\ell)}$, we get:
$\sum_{i \in [k]}{n^{OPT}_i} = \sum_{\ell=0}^{m-1}{f_\ell + g_\ell}
= \sum_{\ell=0}^{m-1}{2 \max(f_\ell,g_\ell) - |f_\ell - g_\ell|}
\Rightarrow n^{SM}_m(P) = \frac{1}{2} \sum_{i \in [k]}{n^{OPT}_i} + \frac{1}{2} \sum_{\ell=0}^{m-1}{|f_\ell - g_\ell|}
\Rightarrow
\lim_{W \to \infty}{\frac{\expectExplicit{n^{SM}_m(P)}{P \sim A}}{mk}} =
\frac{1}{6} + 
\lim_{m \to \infty}{\frac{ \sum_{\ell=0}^{m-1}{\expect{|f_\ell - g_\ell|}}}{2mk}}
$, where the $\frac{1}{6}$ limit has been argued in the proof of Lemma~\ref{lemma_asymptotic_rough}. Then it follows by Lemma~\ref{upper_bound_per_level_with_random_walk} (the inequality), and by Lemma~\ref{lemma_upper_bound_per_level_with_random_walk} (the equality):
$\lim_{m \to \infty}{\frac{\sum_{\ell=0}^{m-1}{\expect{|f_\ell - g_\ell|}}}{2km}}
\le \lim_{m \to \infty}{\frac{m + k + \sum_{\ell=0}^{m-1}{\randwalk{p_\ell}{k-1}}}{2km}}
= \frac{1 + \randwalk{\frac{1}{6}}{k-1})}{2k}
$.
We conclude that $L(k) \le \frac{1}{6} + c(k)$ for $c(k) \equiv \frac{1 + \randwalk{\frac{1}{6}}{k-1}}{2k}$. Finally by Lemma~\ref{lemma_random_walk_with_skip_convert_to_fewer_steps_no_skip}, $\lim_{k \to \infty}{\sqrt{6\pi k} \cdot c(k)} = 1$, and in particular, $\lim_{k \to \infty}{L(k)} = \frac{1}{6}$.
\end{proof}

\begin{remark}
\label{remark_value_of_bound_per_k}
\newtext{For large values of $k$ we have the approximation $c(k) \approx \frac{1}{\sqrt{6 \pi k}}$. For small values of $k$ we can compute directly $RW(\frac{1}{6},k-1)$. In fact, it may be more accurate to bound $L(k)$ from above with the expression $c'(k) = \frac{\randwalk{\frac{1}{6}}{k}}{2k}$ instead of $c(k) = \frac{1 + \randwalk{\frac{1}{6}}{k-1}}{2k}$. In other words, not to regard the $k$th step in any special way, although the sum of all $k$ numbers is $0$ modulo $2^m$. The justification to somewhat ignore this dependency is because $SM$ does not maintain a sum that is a power of $2$: in every level of signed-bits we just zero the signed-bits. This results in reduced correlation the higher we go in the levels. As a concrete toy-example, notice that in level $\ell = 0$ it must be that $f_0 + g_0$ (Definition~\ref{definition_f_g_counters}) is even. However, if $f_0 = g_0 = 1$ then $f_1 + g_1$ is even, and if $f_0 = 2, g_0 = 0$ then $f_1 + g_1$ is odd.}
\end{remark}

\subsection{\newtext{A note regarding Unordered Random Partition}}
\newtext{All of our average-case analysis assumed that we sample ordered-partitions. In mathematics partitions are more commonly studied unordered, whereas ordered-partitions usually appear in combinatorics, in problems of ``throwing balls into bins''. Either way, the choice of either sampling space is a bit arbitrary.}

\newtext{One slight advantage of working with ordered partitions is that they are easier to sample (for Section~\ref{section_experiments}), compared to sampling uniformly (unordered) partitions with a fixed number of parts~\cite{partitionsMathFast}. See~\cite[Section 6.1]{SadehThesis} for more details and sampling-references. However, the main issue is not sampling, but the reduced symmetry. Ordered partitions have much more symmetry because many partitions are equal up-to pemuting their parts. This symmetry let us state and prove Lemma~\ref{lemma_approx_bins_by_rand_bits} in relative ease.}

\newtext{If we try to prove an analogue lemma for (unordered) partitions, we can use a simple reduction to reuse most of the arguments. Noting that the probability of sampling a partition $P$ with $k$ parts is at most $k!$ larger among unordered partitions compared to its probability among ordered partitions (partitions with many equal parts are under-represented in the universe of ordered partitions), we can repeat the arguments to get the exact same expression except that the factor $1 - k^{2-d}$ is replaced by $1 - (k^{2-d})\cdot k!$.}

\newtext{If we still wish the density of the nice set to be $\Theta(1)$, i.e.\ to have $1 - (k^{2-d})\cdot k! \approx \Theta(1)$, we need $d -2 \approx \log_k{(k!)}$. This roughly implies that $d \approx k - \frac{k}{\ln k}$, and since we have $W - d \cdot \lg k$, then $k \cdot \lg k \le W$ is required. Compare this to $\lg k \le W$ in the scenario of ordered partitions.}

\newtext{In order to get a more meaningful result, one may analyze directly a more favorable construction of a nice set $A$. Alternatively, perhaps the existing construction is not as bad as losing a factor of $k!$, because both $A$ and its complement are affected, so perhaps a more careful analysis will yield a better density with respect to the universe of (unordered) partitions.}

\section{Experimental Results}
\label{section_experiments}

\newtext{In this section we list several experiments which we used to complement our theoretical analysis. We divide them into subsections. In Subsection~\ref{subsection_experiments_signed_bits_bounds} we evaluate how good, on average, are the lower and upper bounds of Theorem~\ref{theorem_bounds_signed_bits}. In Subsection~\ref{subsection_experiments_average_case} we evaluate $\expect{\lambda(P)}$ when the partition is sampled uniformly over the set of ordered-partition with fixed parameters $k$ and $W$. In Subsection~\ref{section_experiments_real_data} we use real data to derive partitions, and evaluate $\lambda(P)$ for ``typical partitions'' rather than uniformly random ordered-partitions.}

\subsection{\newtext{Evaluating the Signed-bits LPM Bounds}}
\label{subsection_experiments_signed_bits_bounds}
In this section we evaluate experimentally the upper and lower bounds of Theorem~\ref{theorem_bounds_signed_bits}.
While these signed-bits bounds are partition-specific,
the natural parameters with which we are working are $k$ (number of parts) and $W$ (sum of $2^W$). For all partitions with the same $k$ and $W$, we estimated
the average  ratio of our signed-bits  bounds and the true size of the smallest LPM TCAM.

\looseness=-1
We sample many ordered-partitions for fixed $k$ and $W$, and average over all of them. An ordered partition is a partition where the order of the parts matters, e.g. $[1,3] \ne [3,1]$. We sample uniformly ordered-partitions $P$ as follows: Choose uniformly a subset of $k-1$ values $B \subset \{1,\ldots,2^W-1\}$, denote  the $i^{th}$ smallest value in $B$ by $b_i$ and define $b_0 = 0$ and $b_k = 2^W$. The $i^{th}$ part of the partition is $b_{i} - b_{i-1} > 0$.

For each partition, we estimate how good the bounds are by computing the ratios $\frac{\ceil{(S(P)+1)/2}}{\lambda(P)}$ and $\frac{S(P) + 1 - M(P)}{\lambda(P)}$. Then, we compute the (numeric) expectation of these values by averaging over $10,\!000$ sampled partitions per pair of the parameters $(k,W)$. 

Fig.~\ref{figure_bounds} plots the results for $W \in [10,100]$ and $k=3,4,8,16,100$. Overall, there are 10 graphs in the figure, upper and lower bounds for each of the values of $k$. The upper bounds are separated clearly, while the lower bounds are relatively closer together, and approximately satisfy $\expect{\frac{\ceil{(S(P)+1)/2}}{\lambda(P)}} \approx 0.9$ when $W \ge 20$.

From this experiment, it is clear that the lower bound can be used as a good estimator to the actual value $\lambda(P)$ if $W$ is not too small.
Note that it is much easier to compute this lower bound rather than to compute $\lambda(P)$ itself since it only requires counting bits in the signed-bits representation of $P$'s parts.\footnote{Counting is simpler to implement than Bit Matcher or Niagara, though in practice an implementation would be available in order to compute the TCAM rules. Counting bits is also technically quicker (negligible in practice).} The upper bound, on the other hand, gets looser when $k$ grows larger, which is expected by Remark~\ref{remark_improved_phi_upper_bound}, yet the ratio between these bounds will never exceed $2$ \newtext{(as noted in Remark~\ref{remark_approximation}).}

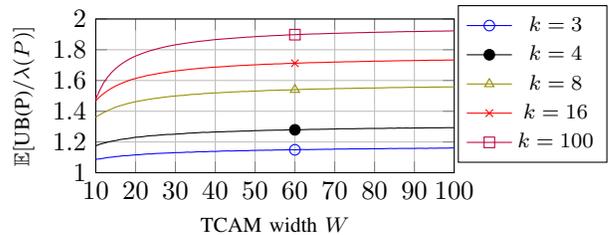
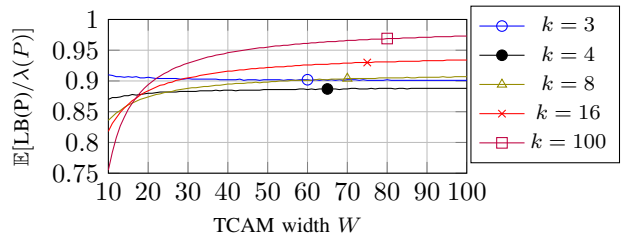
\begin{figure}[t!]
    \centering
    \subfigure[Upper Bound, $UB(P) \equiv S(P)+1 - M(P)$]{
	\begin{tikzpicture}
            \begin{axis}[
                width=0.35\textwidth, height=0.2\textwidth, 
        		xmin=10, xmax=100, xlabel=TCAM width $W$,
                xtick={10,20,30,40,50,60,70,80,90,100},
                ytick = {1.0,1.2,1.4,1.6,1.8,2.0}, ymin = 1.0, ymax = 2.0,
        		ylabel= {$\expect{\text{UB(P)}/\lambda(P)}$},
        		legend style={at={(1.01,1.09)}, anchor=north west,font=\footnotesize, },
        		label style={font=\footnotesize},
        		grid=both
            ]
            
            \addplot [blue,mark=o] coordinates { (60, 1.149) };
            \addlegendentry{$k=3$}
            \addplot [black,mark=*] coordinates { (60, 1.279) }; \addlegendentry{$k=4$}
            \addplot [olive,mark=triangle] coordinates { (60, 1.540) }; \addlegendentry{$k=8$}
            \addplot [red,mark=x] coordinates { (60, 1.712) };
        	\addlegendentry{$k=16$}
        	\addplot [purple,mark=square] coordinates { (60, 1.898) };
            \addlegendentry{$k=100$}

            \addplot [blue] coordinates {
            (10, 1.087) (11, 1.090) (12, 1.094) (13, 1.099) (14, 1.101) (15, 1.105) (16, 1.107) (17, 1.111) (18, 1.111) (19, 1.114) (20, 1.116) (21, 1.117) (22, 1.120) (23, 1.122) (24, 1.122) (25, 1.125) (26, 1.126) (27, 1.127) (28, 1.128) (29, 1.129) (30, 1.131) (31, 1.131) (32, 1.132) (33, 1.134) (34, 1.133) (35, 1.135) (36, 1.137) (37, 1.136) (38, 1.137) (39, 1.139) (40, 1.138) (41, 1.140) (42, 1.141) (43, 1.142) (44, 1.141) (45, 1.142) (46, 1.143) (47, 1.143) (48, 1.143) (49, 1.145) (50, 1.145) (51, 1.146) (52, 1.146) (53, 1.147) (54, 1.146) (55, 1.148) (56, 1.148) (57, 1.148) (58, 1.149) (59, 1.150) (60, 1.149) (61, 1.150) (62, 1.151) (63, 1.150) (64, 1.151) (65, 1.151) (66, 1.152) (67, 1.152) (68, 1.153) (69, 1.152) (70, 1.154) (71, 1.154) (72, 1.154) (73, 1.154) (74, 1.155) (75, 1.155) (76, 1.155) (77, 1.155) (78, 1.155) (79, 1.156) (80, 1.156) (81, 1.156) (82, 1.156) (83, 1.157) (84, 1.158) (85, 1.157) (86, 1.158) (87, 1.158) (88, 1.158) (89, 1.158) (90, 1.158) (91, 1.159) (92, 1.159) (93, 1.158) (94, 1.160) (95, 1.159) (96, 1.160) (97, 1.160) (98, 1.160) (99, 1.161) (100, 1.161)
            }; 
            
            \addplot [black] coordinates {
            (10, 1.175) (11, 1.185) (12, 1.192) (13, 1.199) (14, 1.205) (15, 1.209) (16, 1.215) (17, 1.219) (18, 1.222) (19, 1.226) (20, 1.231) (21, 1.232) (22, 1.235) (23, 1.239) (24, 1.241) (25, 1.242) (26, 1.245) (27, 1.247) (28, 1.249) (29, 1.249) (30, 1.251) (31, 1.253) (32, 1.254) (33, 1.257) (34, 1.258) (35, 1.259) (36, 1.261) (37, 1.263) (38, 1.262) (39, 1.264) (40, 1.265) (41, 1.266) (42, 1.266) (43, 1.268) (44, 1.268) (45, 1.269) (46, 1.270) (47, 1.271) (48, 1.271) (49, 1.272) (50, 1.274) (51, 1.274) (52, 1.275) (53, 1.276) (54, 1.276) (55, 1.277) (56, 1.277) (57, 1.277) (58, 1.279) (59, 1.279) (60, 1.279) (61, 1.280) (62, 1.280) (63, 1.281) (64, 1.282) (65, 1.282) (66, 1.282) (67, 1.283) (68, 1.283) (69, 1.284) (70, 1.284) (71, 1.286) (72, 1.286) (73, 1.286) (74, 1.286) (75, 1.286) (76, 1.287) (77, 1.287) (78, 1.288) (79, 1.288) (80, 1.288) (81, 1.288) (82, 1.288) (83, 1.289) (84, 1.290) (85, 1.289) (86, 1.291) (87, 1.290) (88, 1.291) (89, 1.290) (90, 1.291) (91, 1.291) (92, 1.291) (93, 1.292) (94, 1.292) (95, 1.292) (96, 1.292) (97, 1.292) (98, 1.293) (99, 1.293) (100, 1.293)
            }; 
            
            \addplot [olive] coordinates {
            (10, 1.361) (11, 1.377) (12, 1.392) (13, 1.404) (14, 1.415) (15, 1.427) (16, 1.435) (17, 1.442) (18, 1.450) (19, 1.456) (20, 1.462) (21, 1.467) (22, 1.471) (23, 1.476) (24, 1.480) (25, 1.484) (26, 1.487) (27, 1.491) (28, 1.494) (29, 1.497) (30, 1.499) (31, 1.502) (32, 1.503) (33, 1.506) (34, 1.509) (35, 1.510) (36, 1.511) (37, 1.515) (38, 1.516) (39, 1.518) (40, 1.519) (41, 1.520) (42, 1.522) (43, 1.523) (44, 1.525) (45, 1.526) (46, 1.527) (47, 1.529) (48, 1.528) (49, 1.531) (50, 1.531) (51, 1.532) (52, 1.533) (53, 1.535) (54, 1.536) (55, 1.536) (56, 1.537) (57, 1.538) (58, 1.538) (59, 1.539) (60, 1.540) (61, 1.541) (62, 1.542) (63, 1.543) (64, 1.544) (65, 1.544) (66, 1.545) (67, 1.545) (68, 1.545) (69, 1.546) (70, 1.547) (71, 1.548) (72, 1.547) (73, 1.549) (74, 1.548) (75, 1.550) (76, 1.550) (77, 1.550) (78, 1.551) (79, 1.551) (80, 1.551) (81, 1.552) (82, 1.552) (83, 1.553) (84, 1.554) (85, 1.554) (86, 1.554) (87, 1.554) (88, 1.555) (89, 1.556) (90, 1.555) (91, 1.556) (92, 1.556) (93, 1.556) (94, 1.557) (95, 1.557) (96, 1.557) (97, 1.558) (98, 1.558) (99, 1.558) (100, 1.559)
            }; 
            
            \addplot [red] coordinates {
            (10, 1.467) (11, 1.494) (12, 1.515) (13, 1.534) (14, 1.550) (15, 1.564) (16, 1.576) (17, 1.587) (18, 1.597) (19, 1.606) (20, 1.613) (21, 1.620) (22, 1.626) (23, 1.633) (24, 1.637) (25, 1.642) (26, 1.647) (27, 1.651) (28, 1.655) (29, 1.659) (30, 1.662) (31, 1.665) (32, 1.668) (33, 1.671) (34, 1.674) (35, 1.676) (36, 1.679) (37, 1.681) (38, 1.683) (39, 1.685) (40, 1.687) (41, 1.688) (42, 1.690) (43, 1.692) (44, 1.694) (45, 1.695) (46, 1.697) (47, 1.698) (48, 1.699) (49, 1.701) (50, 1.702) (51, 1.703) (52, 1.705) (53, 1.706) (54, 1.706) (55, 1.708) (56, 1.708) (57, 1.709) (58, 1.711) (59, 1.711) (60, 1.712) (61, 1.713) (62, 1.714) (63, 1.715) (64, 1.716) (65, 1.716) (66, 1.717) (67, 1.717) (68, 1.718) (69, 1.719) (70, 1.720) (71, 1.720) (72, 1.721) (73, 1.722) (74, 1.722) (75, 1.722) (76, 1.723) (77, 1.724) (78, 1.724) (79, 1.725) (80, 1.726) (81, 1.726) (82, 1.726) (83, 1.727) (84, 1.727) (85, 1.728) (86, 1.728) (87, 1.729) (88, 1.729) (89, 1.729) (90, 1.730) (91, 1.730) (92, 1.731) (93, 1.731) (94, 1.732) (95, 1.732) (96, 1.732) (97, 1.732) (98, 1.733) (99, 1.733) (100, 1.733)
            }; 
            
            \addplot [purple] coordinates {
            (10, 1.479) (11, 1.536) (12, 1.586) (13, 1.622) (14, 1.653) (15, 1.677) (16, 1.699) (17, 1.717) (18, 1.733) (19, 1.747) (20, 1.759) (21, 1.769) (22, 1.779) (23, 1.788) (24, 1.796) (25, 1.803) (26, 1.810) (27, 1.816) (28, 1.822) (29, 1.827) (30, 1.832) (31, 1.836) (32, 1.840) (33, 1.844) (34, 1.848) (35, 1.851) (36, 1.855) (37, 1.857) (38, 1.860) (39, 1.863) (40, 1.866) (41, 1.868) (42, 1.870) (43, 1.872) (44, 1.874) (45, 1.876) (46, 1.878) (47, 1.880) (48, 1.882) (49, 1.883) (50, 1.885) (51, 1.887) (52, 1.888) (53, 1.889) (54, 1.891) (55, 1.892) (56, 1.893) (57, 1.894) (58, 1.896) (59, 1.897) (60, 1.898) (61, 1.899) (62, 1.900) (63, 1.901) (64, 1.902) (65, 1.903) (66, 1.903) (67, 1.904) (68, 1.905) (69, 1.906) (70, 1.907) (71, 1.907) (72, 1.908) (73, 1.909) (74, 1.909) (75, 1.910) (76, 1.911) (77, 1.911) (78, 1.912) (79, 1.913) (80, 1.913) (81, 1.914) (82, 1.914) (83, 1.915) (84, 1.915) (85, 1.916) (86, 1.917) (87, 1.917) (88, 1.918) (89, 1.918) (90, 1.918) (91, 1.919) (92, 1.919) (93, 1.920) (94, 1.920) (95, 1.921) (96, 1.921) (97, 1.921) (98, 1.922) (99, 1.922) (100, 1.923)
            }; 
        \end{axis}
    \end{tikzpicture}
    
    \label{figure_bounds_upper_bound}
    }
    
    \subfigure[Lower Bound, $LB(P) \equiv \ceil{(S(P)+1)/2}$]{
	\begin{tikzpicture}
            \begin{axis}[
                width=0.35\textwidth, height=0.2\textwidth, 
        		xmin=10, xmax=100, xlabel=TCAM width $W$,
                xtick={10,20,30,40,50,60,70,80,90,100},
                ytick = {0.7,0.75,0.8,0.85,0.9,0.95,1.0}, ymin = 0.75, ymax = 1.0,
        		ylabel= {$\expect{\text{LB(P)}/\lambda(P)}$},
        		legend style={at={(1.01,1.09)}, anchor=north west,font=\footnotesize, },
        		label style={font=\footnotesize},
        		grid=both
            ]
            
            \addplot [blue,mark=o] coordinates { (60, 0.902) }; \addlegendentry{$k=3$}
        	\addplot [black,mark=*] coordinates { (65, 0.887) }; \addlegendentry{$k=4$}
        	\addplot [olive,mark=triangle] coordinates { (70, 0.904) }; \addlegendentry{$k=8$}
        	\addplot [red,mark=x] coordinates { (75, 0.930) }; \addlegendentry{$k=16$}
        	\addplot [purple,mark=square] coordinates { (80, 0.969) }; \addlegendentry{$k=100$}

            \addplot [blue] coordinates {
            (10, 0.910) (11, 0.909) (12, 0.907) (13, 0.907) (14, 0.906) (15, 0.907) (16, 0.906) (17, 0.907) (18, 0.905) (19, 0.905) (20, 0.906) (21, 0.905) (22, 0.904) (23, 0.904) (24, 0.904) (25, 0.904) (26, 0.904) (27, 0.904) (28, 0.903) (29, 0.903) (30, 0.904) (31, 0.903) (32, 0.903) (33, 0.903) (34, 0.903) (35, 0.903) (36, 0.903) (37, 0.903) (38, 0.902) (39, 0.902) (40, 0.902) (41, 0.902) (42, 0.902) (43, 0.903) (44, 0.901) (45, 0.902) (46, 0.902) (47, 0.902) (48, 0.901) (49, 0.902) (50, 0.902) (51, 0.902) (52, 0.902) (53, 0.902) (54, 0.901) (55, 0.902) (56, 0.902) (57, 0.901) (58, 0.902) (59, 0.902) (60, 0.902) (61, 0.901) (62, 0.902) (63, 0.901) (64, 0.902) (65, 0.901) (66, 0.902) (67, 0.901) (68, 0.902) (69, 0.901) (70, 0.902) (71, 0.902) (72, 0.902) (73, 0.901) (74, 0.902) (75, 0.901) (76, 0.901) (77, 0.901) (78, 0.901) (79, 0.901) (80, 0.902) (81, 0.901) (82, 0.901) (83, 0.901) (84, 0.901) (85, 0.901) (86, 0.901) (87, 0.901) (88, 0.901) (89, 0.901) (90, 0.901) (91, 0.901) (92, 0.901) (93, 0.901) (94, 0.901) (95, 0.901) (96, 0.901) (97, 0.901) (98, 0.901) (99, 0.901) (100, 0.901)
            }; 
            
        	\addplot [black] coordinates {
        	(10, 0.870) (11, 0.873) (12, 0.873) (13, 0.874) (14, 0.876) (15, 0.876) (16, 0.878) (17, 0.877) (18, 0.878) (19, 0.879) (20, 0.880) (21, 0.880) (22, 0.881) (23, 0.882) (24, 0.882) (25, 0.882) (26, 0.882) (27, 0.882) (28, 0.883) (29, 0.883) (30, 0.883) (31, 0.883) (32, 0.883) (33, 0.884) (34, 0.884) (35, 0.884) (36, 0.884) (37, 0.885) (38, 0.884) (39, 0.885) (40, 0.885) (41, 0.885) (42, 0.884) (43, 0.885) (44, 0.885) (45, 0.885) (46, 0.885) (47, 0.885) (48, 0.885) (49, 0.886) (50, 0.886) (51, 0.886) (52, 0.886) (53, 0.886) (54, 0.886) (55, 0.886) (56, 0.886) (57, 0.886) (58, 0.887) (59, 0.887) (60, 0.886) (61, 0.887) (62, 0.886) (63, 0.887) (64, 0.887) (65, 0.887) (66, 0.887) (67, 0.887) (68, 0.886) (69, 0.887) (70, 0.887) (71, 0.888) (72, 0.888) (73, 0.887) (74, 0.887) (75, 0.887) (76, 0.888) (77, 0.887) (78, 0.888) (79, 0.888) (80, 0.888) (81, 0.888) (82, 0.887) (83, 0.888) (84, 0.888) (85, 0.887) (86, 0.888) (87, 0.888) (88, 0.888) (89, 0.888) (90, 0.888) (91, 0.888) (92, 0.888) (93, 0.888) (94, 0.888) (95, 0.888) (96, 0.888) (97, 0.888) (98, 0.888) (99, 0.888) (100, 0.888)
            }; 
            
            \addplot [olive] coordinates {
            (10, 0.836) (11, 0.842) (12, 0.848) (13, 0.852) (14, 0.857) (15, 0.861) (16, 0.864) (17, 0.867) (18, 0.870) (19, 0.872) (20, 0.874) (21, 0.876) (22, 0.878) (23, 0.879) (24, 0.881) (25, 0.882) (26, 0.883) (27, 0.885) (28, 0.886) (29, 0.887) (30, 0.888) (31, 0.889) (32, 0.889) (33, 0.890) (34, 0.891) (35, 0.891) (36, 0.892) (37, 0.893) (38, 0.893) (39, 0.894) (40, 0.895) (41, 0.895) (42, 0.896) (43, 0.896) (44, 0.896) (45, 0.897) (46, 0.897) (47, 0.898) (48, 0.897) (49, 0.898) (50, 0.898) (51, 0.899) (52, 0.899) (53, 0.900) (54, 0.900) (55, 0.900) (56, 0.901) (57, 0.901) (58, 0.901) (59, 0.901) (60, 0.901) (61, 0.901) (62, 0.902) (63, 0.903) (64, 0.902) (65, 0.902) (66, 0.903) (67, 0.903) (68, 0.903) (69, 0.903) (70, 0.904) (71, 0.904) (72, 0.904) (73, 0.904) (74, 0.904) (75, 0.904) (76, 0.904) (77, 0.904) (78, 0.905) (79, 0.905) (80, 0.905) (81, 0.905) (82, 0.905) (83, 0.905) (84, 0.905) (85, 0.906) (86, 0.905) (87, 0.905) (88, 0.906) (89, 0.906) (90, 0.906) (91, 0.906) (92, 0.906) (93, 0.906) (94, 0.907) (95, 0.906) (96, 0.906) (97, 0.907) (98, 0.907) (99, 0.907) (100, 0.907)
            }; 
            
            \addplot [red] coordinates {
            (10, 0.818) (11, 0.831) (12, 0.840) (13, 0.849) (14, 0.856) (15, 0.862) (16, 0.868) (17, 0.872) (18, 0.877) (19, 0.881) (20, 0.884) (21, 0.887) (22, 0.890) (23, 0.893) (24, 0.894) (25, 0.897) (26, 0.898) (27, 0.901) (28, 0.902) (29, 0.904) (30, 0.905) (31, 0.906) (32, 0.907) (33, 0.909) (34, 0.910) (35, 0.911) (36, 0.912) (37, 0.913) (38, 0.914) (39, 0.915) (40, 0.916) (41, 0.916) (42, 0.917) (43, 0.918) (44, 0.918) (45, 0.919) (46, 0.920) (47, 0.920) (48, 0.921) (49, 0.921) (50, 0.922) (51, 0.922) (52, 0.923) (53, 0.923) (54, 0.923) (55, 0.924) (56, 0.924) (57, 0.925) (58, 0.925) (59, 0.926) (60, 0.926) (61, 0.926) (62, 0.927) (63, 0.927) (64, 0.927) (65, 0.927) (66, 0.928) (67, 0.928) (68, 0.928) (69, 0.929) (70, 0.929) (71, 0.929) (72, 0.930) (73, 0.930) (74, 0.930) (75, 0.930) (76, 0.930) (77, 0.931) (78, 0.931) (79, 0.931) (80, 0.931) (81, 0.931) (82, 0.931) (83, 0.932) (84, 0.932) (85, 0.932) (86, 0.932) (87, 0.932) (88, 0.932) (89, 0.933) (90, 0.933) (91, 0.933) (92, 0.933) (93, 0.933) (94, 0.933) (95, 0.934) (96, 0.934) (97, 0.934) (98, 0.934) (99, 0.934) (100, 0.934)
            }; 
            
            \addplot [purple] coordinates {
            (10, 0.755) (11, 0.784) (12, 0.808) (13, 0.827) (14, 0.841) (15, 0.854) (16, 0.864) (17, 0.873) (18, 0.881) (19, 0.888) (20, 0.894) (21, 0.899) (22, 0.904) (23, 0.908) (24, 0.912) (25, 0.915) (26, 0.919) (27, 0.922) (28, 0.924) (29, 0.927) (30, 0.929) (31, 0.931) (32, 0.933) (33, 0.935) (34, 0.937) (35, 0.939) (36, 0.940) (37, 0.942) (38, 0.943) (39, 0.944) (40, 0.946) (41, 0.947) (42, 0.948) (43, 0.949) (44, 0.950) (45, 0.951) (46, 0.952) (47, 0.953) (48, 0.954) (49, 0.954) (50, 0.955) (51, 0.956) (52, 0.957) (53, 0.957) (54, 0.958) (55, 0.958) (56, 0.959) (57, 0.960) (58, 0.960) (59, 0.961) (60, 0.961) (61, 0.962) (62, 0.962) (63, 0.963) (64, 0.963) (65, 0.964) (66, 0.964) (67, 0.964) (68, 0.965) (69, 0.965) (70, 0.966) (71, 0.966) (72, 0.966) (73, 0.967) (74, 0.967) (75, 0.967) (76, 0.968) (77, 0.968) (78, 0.968) (79, 0.968) (80, 0.969) (81, 0.969) (82, 0.969) (83, 0.969) (84, 0.970) (85, 0.970) (86, 0.970) (87, 0.970) (88, 0.971) (89, 0.971) (90, 0.971) (91, 0.971) (92, 0.972) (93, 0.972) (94, 0.972) (95, 0.972) (96, 0.972) (97, 0.973) (98, 0.973) (99, 0.973) (100, 0.973)
            }; 
        \end{axis}
    \end{tikzpicture}
    
    \label{figure_bounds_lower_bound}
    }

    \caption{Estimating the bounds of Theorem~\ref{theorem_bounds_signed_bits}.}
    \label{figure_bounds}
\end{figure}

\subsection{\newtext{Average-case Evaluation}}
\label{subsection_experiments_average_case}

\newtext{In this subsection we evaluate numerically $\expect{\lambda(P)}$ for partitions that are sampled with a fixed pair of parameters $(k,W)$. Remark~\ref{remark_size_lower_bound_k} is a very weak lower bound, and Theorem~\ref{theorem_size_upper_bound_third} is a worst-case upper bound so we can hope to find that on average the size of the TCAM is smaller. In particular, it would also be interesting to see how this expected value behaves compared to the asymptotic bounds that we proved in Theorem~\ref{theorem_avg_case_rand_bits_rough}.}

\newtext{We generated the data by sampling $10,000$ partitions for each $(k,W)$ combination. Fig.~\ref{figure_complexity_vs_W_python_1overkW} shows $\frac{\expect{\lambda(P)}}{kW}$ as a function of the TCAM width $W \in [10,100]$ for partitions of $k=3,4,8,16,100$ targets. One  can think of  $\frac{\expect{\lambda(P)}}{kW}$  as "average rule per bit" since there are $kW$ bits in the binary representation of a partition $P$ with $k$ parts, each a $W$-bit word.}

\newtext{We see that for fixed $k$, $\frac{\expect{\lambda(P)}}{kW}$ converge as $W$ increases. When $W$ is small, there are two opposite phenomena that are not "smoothed out":}
\begin{enumerate}[label=(\arabic*),nosep,topsep=0pt]
    \item \label{effect_1} \newtext{In the lower levels there are more transactions than expected. The reason is that less $1$-bits get cancelled due to carry from lower levels, so Bit Matcher executes relatively more transactions.}
    
    \item \label{effect_2} \newtext{Recall that the top $\lg k$ bit-levels are sparse (mostly zero, as argued in the proof of Theorem~\ref{theorem_size_upper_bound_third}). When $W$ is small the relative part of these levels is larger.}
\end{enumerate}
\newtext{In Fig.~\ref{figure_complexity_vs_W_python_1overkW}, effect~\ref{effect_1} is dominant for $k=3,4$, and effect~\ref{effect_2} is dominant for $k=100$. For $k=8,16$ the effects mostly cancel out.}

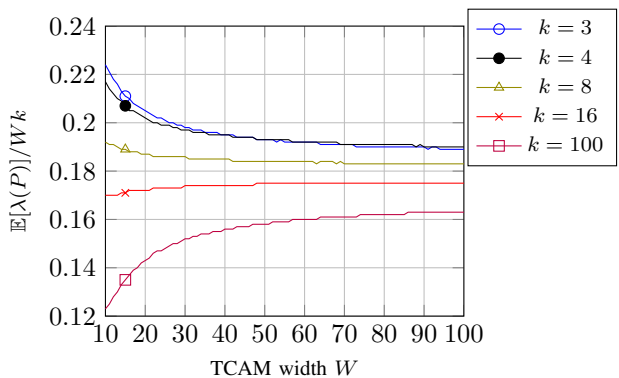
\begin{figure}[!t]
	\centering
	\begin{tikzpicture}
            \begin{axis}[
                width=0.35\textwidth, height=0.3\textwidth, 
        		xmin=10, xmax=100, xlabel=TCAM width $W$,
                xtick={10,20,30,40,50,60,70,80,90,100},
                ytick = {0.12,0.14,0.16,0.18,0.2,0.22,0.24}, ymin = 0.12, ymax = 0.24,
        		ylabel= {$\expect{\lambda(P)}/Wk$},
        		legend style={at={(1.01,1.06)}, anchor=north west,font=\footnotesize, },
        		label style={font=\footnotesize},
        		grid=both
            ]
            
            \addplot [blue,mark=o] coordinates { (15, 0.211) }; \addlegendentry{$k=3$}
        	\addplot [black,mark=*] coordinates { (15, 0.207) }; \addlegendentry{$k=4$}
        	\addplot [olive,mark=triangle] coordinates { (15, 0.189) }; \addlegendentry{$k=8$}
        	\addplot [red,mark=x] coordinates { (15, 0.171) }; \addlegendentry{$k=16$}
            \addplot [purple,mark=square] coordinates { (15, 0.135) }; \addlegendentry{$k=100$}
            
            \addplot [blue] coordinates {
            (10, 0.224) (11, 0.221) (12, 0.218) (13, 0.216) (14, 0.213) (15, 0.211) (16, 0.210) (17, 0.208) (18, 0.207) (19, 0.206) (20, 0.205) (21, 0.204) (22, 0.203) (23, 0.202) (24, 0.202) (25, 0.201) (26, 0.200) (27, 0.200) (28, 0.199) (29, 0.199) (30, 0.198) (31, 0.198) (32, 0.197) (33, 0.197) (34, 0.197) (35, 0.196) (36, 0.196) (37, 0.196) (38, 0.196) (39, 0.195) (40, 0.195) (41, 0.195) (42, 0.195) (43, 0.194) (44, 0.194) (45, 0.194) (46, 0.194) (47, 0.194) (48, 0.193) (49, 0.193) (50, 0.193) (51, 0.193) (52, 0.193) (53, 0.192) (54, 0.193) (55, 0.192) (56, 0.192) (57, 0.192) (58, 0.192) (59, 0.192) (60, 0.192) (61, 0.192) (62, 0.192) (63, 0.192) (64, 0.191) (65, 0.191) (66, 0.191) (67, 0.191) (68, 0.191) (69, 0.191) (70, 0.191) (71, 0.191) (72, 0.191) (73, 0.190) (74, 0.190) (75, 0.190) (76, 0.190) (77, 0.190) (78, 0.190) (79, 0.190) (80, 0.190) (81, 0.190) (82, 0.190) (83, 0.190) (84, 0.190) (85, 0.190) (86, 0.190) (87, 0.190) (88, 0.190) (89, 0.190) (90, 0.190) (91, 0.189) (92, 0.190) (93, 0.190) (94, 0.189) (95, 0.189) (96, 0.189) (97, 0.189) (98, 0.189) (99, 0.189) (100, 0.189)
            }; 
        	\addplot [black] coordinates {
        	(10, 0.217) (11, 0.214) (12, 0.212) (13, 0.210) (14, 0.209) (15, 0.207) (16, 0.205) (17, 0.205) (18, 0.204) (19, 0.203) (20, 0.202) (21, 0.201) (22, 0.200) (23, 0.200) (24, 0.199) (25, 0.199) (26, 0.199) (27, 0.198) (28, 0.198) (29, 0.197) (30, 0.197) (31, 0.197) (32, 0.196) (33, 0.196) (34, 0.196) (35, 0.196) (36, 0.195) (37, 0.195) (38, 0.195) (39, 0.195) (40, 0.195) (41, 0.195) (42, 0.194) (43, 0.194) (44, 0.194) (45, 0.194) (46, 0.194) (47, 0.194) (48, 0.193) (49, 0.193) (50, 0.193) (51, 0.193) (52, 0.193) (53, 0.193) (54, 0.193) (55, 0.193) (56, 0.193) (57, 0.192) (58, 0.192) (59, 0.192) (60, 0.192) (61, 0.192) (62, 0.192) (63, 0.192) (64, 0.192) (65, 0.192) (66, 0.192) (67, 0.192) (68, 0.192) (69, 0.191) (70, 0.191) (71, 0.191) (72, 0.191) (73, 0.191) (74, 0.191) (75, 0.191) (76, 0.191) (77, 0.191) (78, 0.191) (79, 0.191) (80, 0.191) (81, 0.191) (82, 0.191) (83, 0.191) (84, 0.191) (85, 0.191) (86, 0.191) (87, 0.191) (88, 0.190) (89, 0.191) (90, 0.190) (91, 0.190) (92, 0.190) (93, 0.190) (94, 0.190) (95, 0.190) (96, 0.190) (97, 0.190) (98, 0.190) (99, 0.190) (100, 0.190)
            }; 
        	\addplot [olive] coordinates {
        	(10, 0.192) (11, 0.191) (12, 0.191) (13, 0.190) (14, 0.189) (15, 0.189) (16, 0.188) (17, 0.188) (18, 0.188) (19, 0.187) (20, 0.187) (21, 0.187) (22, 0.186) (23, 0.186) (24, 0.186) (25, 0.186) (26, 0.186) (27, 0.186) (28, 0.186) (29, 0.186) (30, 0.186) (31, 0.185) (32, 0.185) (33, 0.185) (34, 0.185) (35, 0.185) (36, 0.185) (37, 0.185) (38, 0.185) (39, 0.185) (40, 0.185) (41, 0.185) (42, 0.184) (43, 0.184) (44, 0.184) (45, 0.184) (46, 0.184) (47, 0.184) (48, 0.184) (49, 0.184) (50, 0.184) (51, 0.184) (52, 0.184) (53, 0.184) (54, 0.184) (55, 0.184) (56, 0.184) (57, 0.184) (58, 0.184) (59, 0.184) (60, 0.184) (61, 0.184) (62, 0.184) (63, 0.183) (64, 0.184) (65, 0.184) (66, 0.184) (67, 0.184) (68, 0.183) (69, 0.184) (70, 0.183) (71, 0.183) (72, 0.183) (73, 0.183) (74, 0.183) (75, 0.183) (76, 0.183) (77, 0.183) (78, 0.183) (79, 0.183) (80, 0.183) (81, 0.183) (82, 0.183) (83, 0.183) (84, 0.183) (85, 0.183) (86, 0.183) (87, 0.183) (88, 0.183) (89, 0.183) (90, 0.183) (91, 0.183) (92, 0.183) (93, 0.183) (94, 0.183) (95, 0.183) (96, 0.183) (97, 0.183) (98, 0.183) (99, 0.183) (100, 0.183)
            }; 
        	\addplot [red] coordinates {
        	(10, 0.170) (11, 0.170) (12, 0.170) (13, 0.170) (14, 0.171) (15, 0.171) (16, 0.172) (17, 0.172) (18, 0.172) (19, 0.172) (20, 0.172) (21, 0.172) (22, 0.173) (23, 0.173) (24, 0.173) (25, 0.173) (26, 0.173) (27, 0.173) (28, 0.173) (29, 0.173) (30, 0.174) (31, 0.174) (32, 0.174) (33, 0.174) (34, 0.174) (35, 0.174) (36, 0.174) (37, 0.174) (38, 0.174) (39, 0.174) (40, 0.174) (41, 0.174) (42, 0.174) (43, 0.174) (44, 0.174) (45, 0.174) (46, 0.174) (47, 0.174) (48, 0.175) (49, 0.175) (50, 0.175) (51, 0.175) (52, 0.175) (53, 0.175) (54, 0.175) (55, 0.175) (56, 0.175) (57, 0.175) (58, 0.175) (59, 0.175) (60, 0.175) (61, 0.175) (62, 0.175) (63, 0.175) (64, 0.175) (65, 0.175) (66, 0.175) (67, 0.175) (68, 0.175) (69, 0.175) (70, 0.175) (71, 0.175) (72, 0.175) (73, 0.175) (74, 0.175) (75, 0.175) (76, 0.175) (77, 0.175) (78, 0.175) (79, 0.175) (80, 0.175) (81, 0.175) (82, 0.175) (83, 0.175) (84, 0.175) (85, 0.175) (86, 0.175) (87, 0.175) (88, 0.175) (89, 0.175) (90, 0.175) (91, 0.175) (92, 0.175) (93, 0.175) (94, 0.175) (95, 0.175) (96, 0.175) (97, 0.175) (98, 0.175) (99, 0.175) (100, 0.175)
            }; 
            \addplot [purple] coordinates {
            (10, 0.123) (11, 0.125) (12, 0.128) (13, 0.130) (14, 0.133) (15, 0.135) (16, 0.137) (17, 0.139) (18, 0.140) (19, 0.142) (20, 0.143) (21, 0.144) (22, 0.146) (23, 0.147) (24, 0.147) (25, 0.148) (26, 0.149) (27, 0.150) (28, 0.150) (29, 0.151) (30, 0.152) (31, 0.152) (32, 0.153) (33, 0.153) (34, 0.154) (35, 0.154) (36, 0.154) (37, 0.155) (38, 0.155) (39, 0.155) (40, 0.156) (41, 0.156) (42, 0.156) (43, 0.157) (44, 0.157) (45, 0.157) (46, 0.157) (47, 0.158) (48, 0.158) (49, 0.158) (50, 0.158) (51, 0.158) (52, 0.159) (53, 0.159) (54, 0.159) (55, 0.159) (56, 0.159) (57, 0.160) (58, 0.160) (59, 0.160) (60, 0.160) (61, 0.160) (62, 0.160) (63, 0.160) (64, 0.161) (65, 0.161) (66, 0.161) (67, 0.161) (68, 0.161) (69, 0.161) (70, 0.161) (71, 0.161) (72, 0.161) (73, 0.161) (74, 0.162) (75, 0.162) (76, 0.162) (77, 0.162) (78, 0.162) (79, 0.162) (80, 0.162) (81, 0.162) (82, 0.162) (83, 0.162) (84, 0.162) (85, 0.162) (86, 0.163) (87, 0.163) (88, 0.163) (89, 0.163) (90, 0.163) (91, 0.163) (92, 0.163) (93, 0.163) (94, 0.163) (95, 0.163) (96, 0.163) (97, 0.163) (98, 0.163) (99, 0.163) (100, 0.163)
            }; 
        \end{axis}
    \end{tikzpicture}
    
    \caption{\newtext{Statistics of $\frac{\expect{\lambda(P)}}{Wk}$ for various values of $W$ and fixed $k$. For fixed $k$, this expectation stabilizes when $W$ grows.}}
    \label{figure_complexity_vs_W_python_1overkW}
\end{figure}

\subsection{\newtext{Complexity of ``Real-Data Partitions''}}
\label{section_experiments_real_data}

\newtext{In this section we analyze the number of rules required for ``real data partitions''. As we do not have a data-center of our own with actual details of the relative power of each target (CPU power, memory, etc.), we resorted to consider public captures, from which we derived partitions according to some arbitrary assumptions. We use the data captured and anonymized in \cite{ftp_packet_captures},\footnote{\newtext{Available for download at: https://ee.lbl.gov/anonymized-traces.html}} which is a 10 day traffic of FTP from January 2003, containing 3.2 million packets in 22 thousand connections between 5832 distinct clients to 320 distinct servers.}

\newtext{While it is impossible to know based on the traffic itself whether it is an intended load-balancing or simply clients connect to different servers according to their preferences, or in the case of FTP possibly some files are located on specific servers and not others, for the sake of ``deriving partitions'' we assume that the traffic represents the desired partition of load in few scenarios we describe below. We stress that while this is an arbitrary decision, we can't deduce much from the data without this assumption or a similar one.}

\newtext{Based on this assumption, we sliced the data to windows of one hour each, starting at the time of the first packet, to get a total of 240 time frames. We extracted from each frame three partitions, according to three types of loads on the servers that communicated in that time: (1) the number of unique clients per server (``load balancing sessions''); (2) the number of incoming packets (``load balancing requests''); (3) the number of outgoing bytes (``load balancing data-processing''). Overall, we get 720 partitions with sums that range in $[18,112] \cup \{298\}$ (connections), $[2143,16306]$ (packets) and $[196637,1618691]$ (sent bytes). The number of parts in the partitions varies among $k \in \{4{-}18, 20, 21, 23, 28, 58, 64, 67, 260\}$.}

\newtext{Most of the partitions do not sum to a power of $2$, as should be expected. Therefore we normalize them to a width that is a multiple of $8$, i.e. each partition is normalized to a sum that is a power of $2^8$, and round the values to integers such that the rounded partition is closest in $L_1$-distance to the initial non-integer partition. Then, for each partition $P$ we compute $\lambda(P)$. Fig.~\ref{figure_real_data_all} shows the results, where the $x$-axis is the number of targets $k$ used to approximate a partition and the $y$-axis is $\lambda(P)$. Note that the scatter is roughly clustered in three stripes, corresponding to $W=8,16,24$ (marked with different colors in the scatter-plot). In fact, because of the typical sizes of the system, every partition of type(1) has $W=8$ except for one anomaly with $W=16$, every partition of type(2) has $W=16$, and every partition of type (3) has $W=24$.}

\begin{figure}[t!]
    \centering
    
	\begin{tikzpicture}
            \begin{axis}[
                width=0.5\textwidth, height=0.3\textwidth,
                xlabel=Number of targets $k$,
        		xmin=0, xmax=30,
                xtick={0,5,10,15,20,25,30},
                ymin = 0, ymax =90,
                ytick = {0,10,20,30,40,50,60,70,80,90},
        		ylabel= {$\lambda(P)$},
        		legend style={at={(0.01,0.99)}, anchor=north west,font=\footnotesize, },
        		label style={font=\footnotesize},
        		grid=both
            ]
            
            \addplot [purple,mark=*,only marks,mark size=1pt] coordinates {
            (10,36) (7,26) (13,49) (8,31) (7,30) (6,26) (8,33) (8,34) (8,37) (18,66) (15,59) (13,49) (8,30) (6,23) (9,37) (11,43) (5,23) (10,38) (8,33) (8,33) (5,25) (6,23) (5,20) (5,24) (6,28) (8,34) (6,26) (8,31) (5,21) (11,39) (5,20) (6,27) (12,43) (7,32) (8,35) (6,27) (13,48) (12,44) (5,22) (5,22) (6,29) (9,31) (6,25) (11,39) (5,20) (7,31) (6,26) (5,20) (5,22) (6,27) (14,47) (5,23) (12,39) (5,23) (12,43) (5,26) (4,19) (8,32) (7,31) (9,34) (6,27) (6,24) (9,37) (10,37) (7,29) (7,28) (7,32) (4,18) (8,33) (7,30) (20,64) (9,34) (7,28) (7,29) (5,22) (7,27) (7,27) (7,29) (7,31) (9,37) (12,47) (6,27) (4,20) (6,27) (4,20) (5,22) (8,32) (12,42) (9,39) (15,49) (7,30) (6,24) (6,23) (5,22) (9,40) (7,28) (21,79) (6,23) (8,33) (10,40) (6,27) (6,23) (14,47) (6,21) (5,23) (7,29) (10,39) (11,43) (7,28) (8,30) (11,42) (5,26) (7,29) (6,29) (5,23) (6,29) (4,18) (23,78) (4,20) (6,27) (11,37) (4,20) (5,25) (5,24) (10,43) (6,26) (4,19) (6,28) (8,30) (15,58) (9,34) (8,31) (7,29) (9,36) (6,27) (8,31) (12,46) (9,35) (14,48) (9,35) (7,27) (8,33) (11,42) (14,55) (11,45) (13,46) (11,38) (16,56) (10,39) (10,36) (18,65) (10,38) (11,40) (7,27) (7,30) (9,37) (8,33) (8,32) (11,41) (7,28) (7,29) (10,39) (11,44) (9,39) (6,27) (14,51) (8,33) (18,62) (9,34) (8,31) (10,41) (9,33) (20,70) (15,58) (11,41) (14,53) (9,36) (9,37) (7,30) (6,28) (9,34) (7,28) (9,35) (10,41) (7,29) (9,34) (7,29) (10,41) (8,29) (11,41) (8,31) (9,34) (8,35) (13,50) (11,45) (11,43) (14,54) (10,38) (8,36) (8,34) (9,37) (8,29) (15,51) (10,35) (6,25) (8,30) (11,39) (9,34) (6,24) (8,33) (9,38) (8,33) (8,33) (11,40) (15,51) (13,47)  (16,60) (28,89) (6,25) (14,50) (9,34) (12,47) (13,51) (11,43) (8,35) (12,46) (7,28) (5,22) (6,22) (10,39) (9,33) (9,36) (7,29) (8,30) (14,51) (7,24)
            }; \addlegendentry{bytes partitions.}
            
            \addplot [blue,mark=o,only marks,mark size=1pt] coordinates {
            (10,23) (7,19) (13,33) (8,19) (8,22) (6,18) (8,21) (8,21) (8,21) (18,35) (15,36) (13,24) (8,18) (6,19) (9,22) (11,28) (5,16) (10,20) (8,21) (8,20) (5,15) (6,15) (5,12) (5,15) (6,16) (8,21) (6,18) (7,21) (5,15) (11,27) (5,17) (6,16) (11,24) (10,26) (8,23) (6,17) (13,29) (13,27) (5,15) (5,15) (6,17)  (9,18) (6,15) (12,24) (5,14) (7,19) (6,17) (5,16) (5,17) (6,15) (14,33) (5,12) (12,29) (5,16) (12,28) (5,14) (4,11) (8,20) (7,21) (9,19) (6,17) (6,15) (10,25) (10,27) (7,19) (7,17) (7,19) (4,11) (8,20) (7,18) (20,41) (9,21) (7,18) (7,18) (5,17) (7,16) (7,17) (7,17) (7,21) (9,23) (12,26) (6,19) (4,13) (6,14) (4,11) (5,18) (9,23) (12,27) (9,24) (15,32) (7,16) (6,16) (6,14) (5,16) (9,26) (7,17) (21,45) (6,17) (8,21) (10,22) (7,19) (6,15) (14,32) (6,13) (5,11) (7,16) (10,26) (11,24) (6,15) (8,21) (11,24) (5,15) (7,15) (6,18) (5,14) (6,17) (6,17) (23,51) (4,13) (6,14) (11,26) (4,13) (5,16) (5,14) (10,30) (6,18) (5,11) (6,19) (8,21) (15,34) (8,21) (8,21) (7,19) (9,25) (6,18) (8,24) (12,27) (9,23) (14,34) (10,23) (7,17) (8,22) (11,29) (14,33) (11,26) (13,29) (10,26) (16,35) (10,26) (9,24) (18,43) (10,28) (11,23) (7,15) (7,20) (9,24) (8,19) (8,18) (11,25) (7,22) (7,17) (10,25) (11,27) (9,20) (6,18) (14,34) (8,19) (18,36) (9,21) (8,19) (11,27) (9,20) (20,43) (15,37) (11,28) (14,32) (9,25) (9,25) (7,18) (7,15) (9,21) (7,18) (9,23) (10,25) (7,19) (9,28) (7,19) (10,23) (8,21) (11,27) (8,19) (9,21) (8,25) (13,34) (11,27) (11,25) (14,36) (10,23) (8,21) (8,23) (9,22) (9,19) (15,34) (10,24) (7,16) (8,17) (11,28) (9,24) (6,17) (8,20) (9,19) (9,23) (8,22) (11,30) (15,33) (14,29) (17,39) (28,55) (6,16) (15,31) (10,25) (11,26) (13,34) (10,27) (8,25) (12,24) (7,20) (5,15) (6,16) (10,27) (9,23) (10,24) (7,19) (8,18) (14,30) (7,16) 

            }; \addlegendentry{packets partitions.}
            
            \addplot [black,mark=square,only marks,mark size=1pt] coordinates {
            (10,10) (7,11) (13,16) (8,13) (8,13) (6,10) (8,10) (8,13) (8,14) (18,21) (15,20) (14,20) (8,13) (6,10) (9,14) (11,13) (5,8) (10,16) (8,13) (8,11) (5,8) (6,10) (5,8) (5,9) (6,7) (8,12) (6,11) (8,12) (5,10) (11,17) (5,8) (6,9) (12,19) (11,11) (8,12) (6,8) (13,20) (15,23) (5,8) (5,7) (6,10) (9,14) (6,9) (12,18) (5,9) (7,13) (6,11) (5,7) (5,8) (6,10) (14,23) (5,5) (12,15) (5,8) (12,19) (5,9) (4,7) (8,14) (7,11) (9,15) (6,9) (6,10) (10,15) (10,16) (7,11) (7,11) (7,10) (4,7) (8,11) (7,12) (20,31) (9,14) (7,11) (7,12) (5,9) (7,11) (7,12) (7,11) (7,10) (9,11) (12,12) (6,9) (4,8) (6,10) (4,6) (5,8) (9,14) (12,17) (9,15) (15,23) (7,11) (6,11) (6,10) (5,8) (9,15) (7,10) (21,35) (6,10) (8,12) (10,14) (7,10) (6,12) (14,23) (6,10) (5,8) (7,10) (10,16) (11,16) (7,10) (8,12) (11,15) (5,7) (7,11) (6,11) (5,8) (6,6) (6,8) (23,39) (4,7) (6,11) (11,14) (4,7) (5,7) (5,9) (10,16) (6,9) (5,8) (6,10) (8,13) (15,18) (9,14) (8,10) (7,12) (9,11) (6,10) (8,12) (12,17) (9,15) (14,19) (10,12) (7,11) (8,9) (11,15) (14,16) (11,16) (13,16) (11,14) (16,18) (10,15) (10,15) (18,28) (10,16) (11,19) (7,9) (7,12) (9,13) (8,11) (8,12) (11,14) (7,11) (7,9) (10,15) (11,13) (9,12) (6,9) (14,20) (8,11) (18,22) (9,12) (8,12) (11,14) (9,11) (20,24) (15,17) (11,17) (14,21) (9,13) (9,12) (7,11) (7,11) (9,13) (7,10) (9,15) (10,17) (7,10) (9,12) (7,11) (10,17) (8,13) (11,12) (8,10) (9,11) (8,12) (13,14) (11,12) (11,16) (14,15) (10,14) (8,12) (8,13) (9,14) (9,13) (15,24) (10,16) (7,10) (8,13) (11,18) (9,14) (6,10) (8,12) (9,13) (9,13) (8,12) (11,17) (15,17) (14,17) (17,21) (28,41) (6,9) (15,17) (10,10) (12,17) (13,16) (11,18) (8,11) (12,14) (7,10) (5,8) (6,9) (10,15) (9,15) (10,16) (7,9) (8,12) (14,22) (7,11) 
            }; \addlegendentry{sessions partitions.}
            
        \end{axis}
    \end{tikzpicture}

    \caption{\newtext{The complexity of partitions derived from real-data. As described in Section~\ref{section_experiments_real_data}, we extracted three types of partitions based on the number of bytes sent by the servers (``load''), the number of packets received (``requests''), and the number of unique clients (``sessions''). Each partition was normalized to a sum $2^W$ for $W$ closest from above and a multiple of $8$, and then its complexity $\lambda(P)$ was computed. In this figure, all the bytes-partitions all have $W=24$, all the packets-partitions have $W=16$ and all the sessions-partitions have $W=8$, and their clustering over different trend-lines is noticeable. We omitted from this scatter 4 partitions per class, the partitions with $k=58,64,67,260$, in order to improve the visual resolution, but they fit the trends as well. The linear trends are approximately $y=3.09x+7.56$ for $W=24$ ($R^2 = 0.996$), $y=1.94x+5.05$ ($R^2 = 0.993$) for $W=16$ and $y=1.25x+2.08$ for $W=8$ ($R^2 = 0.953$).}}
    \label{figure_real_data_all}
\end{figure}
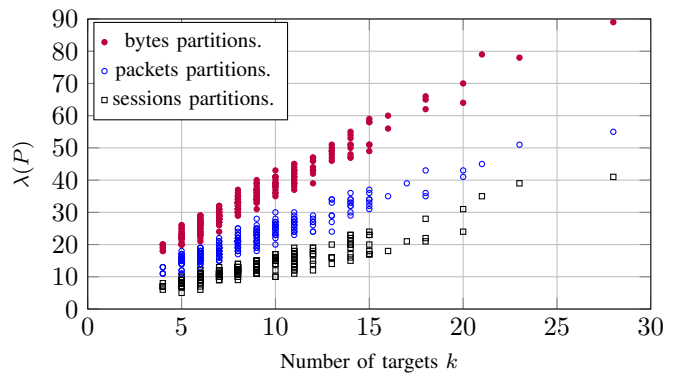

\section{Related Work}
\label{section_related_work}

\newtext{\textbf{Matching-based Implementations:}
The work of \cite{AccurateExp} showed the private case of Theorem~\ref{theorem_tight_signed_bits_bounds} for $k=2$ targets. An earlier work \cite{WangBR11} considered only restricted TCAM encodings in which rules are \emph{disjoint}. For instance, the partition $P=[4,3,1]$ for $W=3$ is implemented with the four rules $(\textsc{0**} \to 1, \textsc{10*} \to 2, \textsc{110} \to 2, \textsc{111} \to 3)$. Since TCAMs allow overlapping rules and resolve overlaps by ordering the rules, this early approach does not take full advantage of them. For example $P$ can also be implemented by prioritizing longer prefix rules as $(\textsc{0**} \to 1, \textsc{111} \to 3, \textsc{1**} \to 2)$.}

\newtext{\textbf{Hashing-based Implementations:}
Hash-based solutions for load-balancing use an array,
each of its cells contains a target. The fraction of the cells containing a particular target determines the fraction of the addresses that this target gets. 
This solution is also known as WCMP~\cite{WCMP, Zegura} or as 
ECMP~\cite{RFC2992} when traffic is split equally. \cite{WuTang} studies the relation between the size of the array and how good it approximates a desired distribution.
While the above works studied a fixed output distribution, in a dynamic scenario mapping has to be updated following a change in the required  distribution. ~\cite{Chao08, Chim, KandulaKSB07} considered such updates for load balancing over multiple paths. They suggested update schemes that reduce transient negative impact of packet reordering.
A recent approach~\cite{DASH_alg} refrains from memory blowup by comparing the hash to range-boundaries. Since the hash is tested sequentially against each range, it restricts the total number of load-balancing targets.}

\newtext{\textbf{Partitions vs. Functions:} This paper studies efficient representations of partitions. A \emph{partition} specifies the number of addresses that have to be mapped to each possible target but leaves the freedom to choose these addresses. In contrast a \emph{function} specifies exactly the target of each address. Note that there may still be multiple ways to implement a function with a TCAM. Finding the smallest list of prefix rules that realizes a given function can be done in polynomial time with dynamic programming~\cite{DravesKSZ99, suri03}.  When we are not restricted to prefix rules the problem is  NP-hard~\cite{McGeerY09}. The particular family of ``range functions'' where the preimage of each target is an interval was carefully studied due to its popularity in packet classifiers for access control~\cite{Gray, ORange}. Going back to implementing partitions, ~\cite{AccurateExp} proved that any partition to two targets has an optimal realization as a range function.}

\newtext{\textbf{Signed-digits Arithmetics:} The paper \cite{SignedBitsFastArithmetics} uses the generalized definition of signed-digits representation to speed-up arithmetic operations. An alternative view of signed-digits representation is representing an integer as the difference of two non-negative integers. \cite{SignedBits1960} uses this alternative framing to investigate binary arithmetic. Overall, signed-bits are of interest since they come up in minimization/optimization scenarios, see \cite[Section 6]{SignedBitsLooplessGrayCode} for a survey. In the context of TCAMs analysis, \cite{AccurateExp} used this representation to give an exact expression for $\lambda(P)$ when $k=2$, prior to our generalization for $k \ge 3$.}

\section{Conclusions} 
\label{section_conclusions}
In this paper we thoroughly studied the size of a minimal LPM TCAM table that implements a specific partition $P$ of the address-space. We proved that a partition requires no more than $\frac{kW}{3}$ rules, and also that a ``typical'' partition should require roughly half of this number of rules, about $\frac{1}{6}kW$ rules, which was also supported by our simulations. While the analysis was done asymptotically for large values of $W$, our simulations show that the results still hold even when $W$ is not too large.

While there exist partitions that have a very compact representation (with only $k$ rules), since the expected LPM TCAM for a partition is $\frac{1}{6}kW$ rules, an effective way to reduce the size of the encoding of an arbitrary partition is by rounding the binary representation of its parts to reduce the effective width from $W$ to $W'$. This may reduce the TCAM size (on average) by a factor of approximately $W' / W$ (from $\frac{1}{6}kW$ to $\frac{1}{6}kW'$). We note that the problem of finding a close partition that requires less rules has been stated and solved in \cite{TCAM_Linf_TON} and \cite{TCAM_L1_INFOCOM}. The results of our analysis provide a rule-of-thumb to estimate how far the compact partition is expected to be from the desired partition.

We also analyzed bounds that depend on the signed-bits representation of a partition, and found in our simulations that the signed-bits lower bound in terms of this representation provides a good estimation to $\lambda(P)$, in expectation. The signed-bits representation of a number has less non-zero coefficients if we round it to a multiple of a large power of $2$ (ignoring least significant (signed-) bits). This again shows that we can reduce the number of TCAM rules by reducing the effective width of the values, by rounding the parts of the desired partition.


\bibliographystyle{IEEEtran}
\bibliography{reference}


\end{document}